\documentclass[12pt, notitlepage, reqno]{article}
 
 \usepackage{natbib}
 \usepackage{titling}
 \usepackage{amssymb, latexsym}
 \usepackage{amsmath} 
 \usepackage{amsthm}
 \usepackage{bm}
 \usepackage{bbm}
 \usepackage[mathscr]{eucal}
 \usepackage{enumerate}
\usepackage{hhline}
\usepackage[table]{xcolor}
\usepackage{multirow}
\usepackage[]{algorithm2e}
\usepackage{mathtools}

\newtheorem{theorem}{Theorem}[section]
\newtheorem{lemma}[theorem]{Lemma}

\newtheorem{corollary}[theorem]{Corollary}

\newtheorem{dfn}[theorem]{Definition}
\newtheorem{condition}{Condition}

\DeclareMathOperator*{\argmin}{arg\,min}
\DeclareMathOperator*{\argmax}{arg\,max}

\DeclarePairedDelimiter\ceil{\lceil}{\rceil}

\newenvironment{brsm}{% % short for 'bracketed small matrix'
  \bigl[ \begin{smallmatrix} }{%
  \end{smallmatrix} \bigr]}
\newcommand{\intd}{\mbox{\textup{d}}} % d at the end of an integral

\newcommand{\setr}{\mathbb{R}} % set of real numbers
\newcommand{\setnn}{\mathbb{N}_{0}} % set of natural numbers
\newcommand{\setnnn}{\mathbb{N}_{1}} % set of nonnegative integers

\providecommand{\keywords}[1]{\textbf{\textit{Key Words:}} #1}
\title{Bayesian Mosaic: Parallelizable Composite Posterior}
\author{Ye Wang and David B. Dunson\\
 Department of Statistical Science \\
Duke University}

\date{}

\begin{document}

\begin{titlepage}
\maketitle
\thispagestyle{empty}

 \begin{abstract}
This paper proposes {\em Bayesian mosaic}, a
parallelizable composite posterior, for scalable Bayesian inference
on a broad class of multivariate discrete data models. 
Sampling is embarrassingly parallel since {\em Bayesian mosaic} is a
multiplication of component posteriors that can be independently sampled from.
Analogous to composite likelihood methods, these component posteriors
are based on univariate or bivariate
marginal densities. Utilizing the fact that the score functions of
these densities are unbiased, we show that {\em Bayesian mosaic}
is consistent and 
asymptotically normal under mild
conditions. Since the evaluation of univariate or bivariate
marginal densities can rely on numerical integration, sampling
from {\em Bayesian mosaic} bypasses the traditional data augmented Markov chain Monte
Carlo (DA-MCMC) method, which has a provably slow mixing rate when data are
imbalanced. Moreover, we show that sampling from {\em Bayesian
  mosaic} has better scalability to large sample size than DA-MCMC.
The method is evaluated via
simulation studies and an application on a citation count dataset.
 \end{abstract}

\end{titlepage}

\setcounter{page}{1}

 \section{INTRODUCTION}
 
 \setlength{\parindent}{5ex}

There is great interest in designing flexible models
for multivariate discrete data. 
A common strategy is to define
a generalized linear model (GLM) for each variable, with dependence in
the different variables induced through including multivariate latent 
variables in the GLMs. 
Alternatively, discrete data can be directly linked to the latent
variables via some link functions.
A popular choice for the latent
variable distribution is the multivariate Gaussian due to simplicity in modeling the
dependence structure.
For instance, multivariate Poisson regression with the underlying intercepts 
modeled jointly as a Gaussian has been widely used in
accident analysis~\citep{ma2008multivariate, el2014investigation}.
\citet{canale2011bayesian} proposed a multivariate count model that
handles both over-dispersion and under-dispersion. This model uses a
 rounding function
to directly link the multivariate count data to a latent Gaussian. We term these models as multivariate latent Gaussian
 models. Unfortunately, despite their great flexibility, the usage of
 this class of models is limited by the computationally challenging model fitting.

The challenge is due partially to the fact that likelihood functions
marginalizing out the latent variables lack analytic forms. Hence, Bayesian inference is usually done via data augmented Markov chain Monte Carlo
 (DA-MCMC) algorithms that sample both the latent variables and the model
 parameters from their joint posterior. However, it is well known that posterior dependence
 between the latent variables and the model parameters can substantially
 slow down the mixing rate of the Markov chain. In fact,
\citet{johndrow2016inefficiency} has shown that the mixing rate can
 be so slow that the DA-MCMC sampler cannot generate any reliable posterior samples when the data
 are severely imbalanced (e.g., excessive zeros in count data).

One possible solution is to bypass sampling entirely using one of the
following two strategies.  The first is the integrated nested Laplace approximation
(INLA), which is designed for latent Gaussian models that have a small
number of parameters remaining after marginalizing out the latent variables~\citep{rue2009approximate}.  Although INLA has had excellent
performance in specialized settings, in many applications, there are moderate to large numbers of
population parameters, ruling out such approaches. 
Another strategy is the so-called variational
approximations~\citep{attias2000variational, jaakkola2000bayesian},
which introduce an approximate posterior with a factorized form. One then optimizes the
parameters of this approximate posterior to minimize its Kullback-Leibler divergence from
the exact posterior. However, in general one has no idea how accurate
this approximation is and additionally it is well known that it
often substantially underestimates the true posterior
covariance.

We propose {\em Bayesian mosaic}, which is a surrogate
posterior derived by multiplying a collection of component
posteriors. Unlike variational approximations, where one has to choose
the variational class and optimize its parameters, the construction of
{\em Bayesian
  mosaic} is automatically determined by the data distribution.
It is related to the composite likelihood
approach~\citep{cox2004note} with its component posteriors being based
on univariate and bivariate marginal
distributions. However,  {\em Bayesian
  mosaic} is different from Bayesian composite
likelihood methods~\citep{pauli2011bayesian} in that it has an
easy-to-sample multiplicative form while a posterior density induced
by a composite likelihood does not. Utilizing that these marginal densities have unbiased score
functions, we have shown that {\em Bayesian mosaic} is consistent and
asymptotically normal under mild conditions. 
It is applicable to a class of  {\em mosaic-type} data distributions
that covers and is much broader than the class of multivariate latent Gaussian models mentioned earlier.

We also propose an efficient parallel
sampling strategy utilizing the posterior dependence structure induced
by the multiplicative form. This parallelization is substantially different from standard parallel
MCMC algorithms which are based on
partitions of the dataset~\citep{wang2013parallelizing,
  scott2016bayes}.  The sparse dependence structure of {\em Bayesian
  mosaic} allows us to  directly sample from each component posterior independently.
Moreover, we have shown that the asymptotic per-iteration computational
complexity of sampling from {\em Bayesian mosaic} is linear in the cardinality of the observed data, which is
in general much smaller than the sample size. On the other
hand, the per-iteration computational
complexity of DA-MCMC is linear in the
sample size.

The remainder of the paper is organized as follows. \S\ref{sec:bm}
provides definitions
and the sampling strategy. \S\ref{sec:theory}  provides theories on the richness of the  {\em
  mosaic-type} class and asymptotic properties of {\em
  Bayesian mosaic}. The
performance is demonstrated via both simulation studies and an
application on a
citation network dataset in
\S\ref{sec:experiment}.

\section{Bayesian Mosaic} \label{sec:bm}
We start by introducing
notation that will be used throughout the paper. Before presenting
the formal definition, we will first motivate the
proposed method by introducing a class of multivariate latent Gaussian
 models and describing computational issues
DA-MCMC algorithms encounter in fitting these models. After
defining {\em Bayesian mosaic}, we present a sampling algorithm
and a post-processing method to handle parameter constraints. We
end the section by generalizing {\em Bayesian mosaic} to
dependent data.

\subsection{Notations} \label{ssec:notation}
We represent vectors by
lower case letters and matrices by
capital letters, both in a boldface font. Unless otherwise stated, all
vectors will be column vectors. We use $\setr$ to denote the set of
all real numbers, $\setnn$ the nonnegative integers, $\setnnn$ the positive integers and $\|\cdot\|$ the Euclidean
norm. For $d\in \setnnn$, $\bm{\theta}\in\setr^{d}$ and $\delta>0$, we
define a radius-$\delta$ ball of $\bm{\theta}$ at $\bm{\theta}_{0}$ as
$\mathcal{B}_{\bm{\theta}}(\bm{\theta}_{0},\delta)=\{\bm{\theta}:\|\bm{\theta}-\bm{\theta}_{0}\|<\delta\}$. For
succinctness, we denote the multiple integral of a multivariate
function $g(\bm{y})$ as
\begin{align*}
\int g(\bm{y}) \intd \bm{y} = \int\cdots\int g(y_1,\ldots,y_p) \intd
  y_1\cdots \intd y_p.
\end{align*}

We will always use $f$ to denote a density function and $\ell$ to
denote a log-density function. The density function will be presented
in a conditional style, e.g., $f(\bm{y}|\bm{\theta})$, where
$\bm{\theta}$ are the model parameters. Given that $\bm{y}$ follows
some distribution $P_{\bm{\theta}}$ with a density function
$f(\bm{y}|\bm{\theta})$, we use $\mathbb{E}_{\bm{\theta}}g(\bm{y})$ to denote the expectation of
$g(\bm{y})$. More specifically,
\begin{align*}
\mathbb{E}_{\bm{\theta}}g(\bm{y}) = \int g(\bm{y}) f(\bm{y}|\bm{\theta}) \intd\bm{y}.
\end{align*}
We use $\phi(\bm{x}|\bm{\mu},\bm{\Sigma})$ to denote the density
function of a Gaussian distribution with mean $\bm{\mu}$ and
covariance $\bm{\Sigma}$, and use
$\Phi\left(U|\bm{\mu},\bm{\Sigma}\right)$ to denote its cdf function:
\begin{align*}
\phi(\bm{x}|\bm{\mu},\bm{\Sigma}) = &
  |2\pi\bm{\Sigma}|^{-1/2}\exp\big\{-(\bm{x}-\bm{\mu})^{\top}\bm{\Sigma}^{-1}(\bm{x}-\bm{\mu})/2\big\},\\
\Phi\left(U|\bm{\mu},\bm{\Sigma}\right) = &
                                            \int_{U}\phi(\bm{x}|\bm{\mu},\bm{\Sigma})\intd  \bm{x}.
\end{align*}

For better representation of the higher-order remainder of the Taylor
expansion for multivariate functions, we adopt the
notations of \cite{folland2005higher}. For any $d\in \setnnn$, a $d$-dimensional {\em multi-index}
$\bm{\alpha}$ for $\bm{x}=(x_1,\ldots,x_d)^{\top}\in\setr^{d}$ is defined as a $d$-tuple of nonnegative integers, i.e., 
$\bm{\alpha} = (\alpha_1, \ldots, \alpha_d)$,
where $\alpha_{j}\in \setnn$ for $j=1,\ldots,d$. We further
define 
\begin{align*}
|\bm{\alpha}| = \sum_{j=1}^{d}\alpha_j, & \quad \bm{\alpha}! =
                                     \prod_{j=1}^{d}\alpha_j!,\quad \bm{x}^{\bm{\alpha}} = \prod_{j=1}^{d}x_j^{\alpha_j} , \quad \partial^{\bm{\alpha}} g(\bm{x}) = \frac{\partial^{|\bm{\alpha}|} g(\bm{x})}{\partial
  x_1^{\alpha_{1}}\cdots\partial x_d^{\alpha_d}}.
\end{align*}

We will also use the following vector calculus notation to ease our
representation of the gradient vector and the Hessian matrix. Considering
two vectors
$\bm{\eta}=\big(\eta_{1},\ldots,\eta_{d_{\eta}}\big)^{\top}$,
$\bm{\zeta}=\big(\zeta_{1},\ldots,\zeta_{d_{\zeta}}\big)^{\top}$ and some function $g(\bm{\eta},\bm{\zeta})$, we denote the gradient of $f(\bm{\eta},\bm{\zeta})$
w.r.t. $\bm{\eta}$ as 
\[\triangledown_{\bm{\eta}}g(\bm{\eta},\bm{\zeta})=\left[\frac{\partial
  g(\bm{\eta},\bm{\zeta})}{\partial \eta_{1}},\ldots, \frac{\partial
  g(\bm{\eta},\bm{\zeta})}{\partial
  \eta_{d_{\eta}}}\right]^{\top},\]
and the gradient of $g(\bm{\eta},\bm{\zeta})$
w.r.t. $\bm{\eta}$ and $\bm{\zeta}$ as
\[\triangledown_{\bm{\eta},\bm{\zeta}}g(\bm{\eta},\bm{\zeta})=\left[\frac{\partial
  g(\bm{\eta},\bm{\zeta})}{\partial \eta_{1}},\ldots, \frac{\partial
  g(\bm{\eta},\bm{\zeta})}{\partial
  \eta_{d_{\eta}}}, \frac{\partial
  g(\bm{\eta},\bm{\zeta})}{\partial \zeta_{1}},\ldots, \frac{\partial
  g(\bm{\eta},\bm{\zeta})}{\partial
  \zeta_{d_{\zeta}}}\right]^{\top}.\]
We
further define
\begin{align*}
\triangledown_{\bm{\zeta}}\triangledown_{\bm{\eta}}g(\bm{\eta},\bm{\zeta})=
\begin{bmatrix}
    \frac{\partial^{2}
  g(\bm{\eta},\bm{\zeta})}{\partial \eta_{1}\partial \zeta_{1}} &\cdots & \frac{\partial^{2}
  g(\bm{\eta},\bm{\zeta})}{\partial
  \eta_{1}\partial \zeta_{d_{\zeta}}} \\
\vdots & & \vdots \\
 \frac{\partial^{2}g(\bm{\eta},\bm{\zeta})}{\partial
  \eta_{d_{\eta}}\partial \zeta_{1}} &\cdots & \frac{\partial^{2}
  g(\bm{\eta},\bm{\zeta})}{\partial
  \eta_{d_{\eta}}\partial \zeta_{d_{\zeta}}}
\end{bmatrix}.
\end{align*}
We
suppress $\triangledown_{\bm{\eta}}\triangledown_{\bm{\eta}}$ to
$\triangledown_{\bm{\eta}}^{2}$. Note that
$\triangledown_{\bm{\eta},\bm{\zeta}}^{2}g(\bm{\eta},\bm{\zeta})$ is
the Hessian of $g(\bm{\eta},\bm{\zeta})$.

Consider
$p\in\setnnn$ and a
sequence indexed by two subscripts $\left\{x_{st}\right\}$ where
$1\leq s < t \leq p$. Whenever we write $x_{12},\ldots,x_{(p-1)p}$,
we mean
\[x_{12},\ldots,x_{1p},x_{23},\ldots,x_{2p},\ldots,x_{(p-2)(p-1)},
  x_{(p-2)p}, x_{(p-1)p}\]
where the elements are ordered in a row-major manner.

\subsection{Multivariate Latent Gaussian
 Model} \label{ssec:mvt-latent-gaussian}
Suppose $p\in\setnnn$, $n\in\setnnn$ and $\bm{y}_1,\ldots,\bm{y}_n$
are i.i.d. $p$-dimensional observations from some multivariate
latent Gaussian model. Letting $\bm{y}=(y_{1},\ldots,y_{p})^{\top}\in\mathcal{Y}\subseteq\setr^{p}$ and introducing
$\bm{x}=(x_{1},\ldots,x_{p})^{\top}\in\setr^{p}$, the density function of
the model can be written as
\begin{align} \label{eq:mvt-lg-pois}
\begin{split}
f(\bm{y}|\bm{\mu},\bm{\Sigma})=\int
  \prod_{j=1}^{p}h_{j}(y_{j}|x_{j})\phi(
\bm{x}|\bm{\mu}, \bm{\Sigma})
\mbox{d}\bm{x},
\end{split}
\end{align}
where $h_{j}$'s are univariate density functions, $\bm{\Sigma}=\left\{\sigma_{st}\right\}$
is a $p\times p$ positive definite matrix and $\bm{\mu}
=\left(\mu_1,\ldots,\mu_{p}\right)^{\top}\in \setr^{p}$. We refer to $h_{j}$'s as
link densities.

The integral in (\ref{eq:mvt-lg-pois}) usually does not have an analytical
solution, and accurate numerical integration is infeasible even for moderately large $p$. Hence, fully Bayesian inference is usually
based on a DA-MCMC algorithm, where $\bm{x}$ is augmented and sampled
together with the model parameters $\bm{\mu}$ and $\bm{\Sigma}$.

If we let $h_{j}$'s be discrete data densities, then
(\ref{eq:mvt-lg-pois}) provides a rich class of
multivariate discrete data models. However, real world discrete
datasets are often severely imbalanced. Taking online advertising as
an example, the click through rate of a certain link is usually very
close to 0. Supposing that one wants to fit a logistic regression to predict
the probability of a certain user's clicking the link, only a tiny
faction of the responses will be 1. 

Unfortunately, DA-MCMC has a provably slow mixing rate in imbalanced
discrete data problems. Considering an intercept-only probit model and
assuming that the data are infinite imbalanced,
\citet{johndrow2016inefficiency} have shown that the step size of
DA-MCMC is
roughly $O(\frac{1}{\sqrt{n}})$, while the width of the high
probability region of the posterior is roughly $O(\frac{1}{\log
  n})$. For large $n$, the step size will become much smaller than the
width of the high probability bulk, 
causing extreme slow
mixing. Moreover, this mismatch will become worse as $n$ grows, and
presents huge practical problems in broad settings.

Another drawback of DA-MCMC is its poor scalability to large sample
size. The per-iteration
computational complexity is at least $O(n)$ due to the need to sample an augmented
$\bm{x}_{i}$ for each $\bm{y}_{i}$. Moreover, the number of model
parameters in (\ref{eq:mvt-lg-pois}) increases quadratically as the
data dimensionality grows.

One way to bypass DA-MCMC is to evaluate the integral in
(\ref{eq:mvt-lg-pois}) directly via deterministic numerical
integration methods.
Unfortunately, these methods are only feasible for small $p$. A potential
solution is to approximate Bayesian inference by using a composite
likelihood whose individual components are low-dimensional conditional
or marginal densities that can be numerically evaluated~\citep{pauli2011bayesian}. {\em
  Bayesian mosaic} is partially motivated by this idea.

Suppose $\bm{y}=\left(y_{1},\ldots,y_{p}\right)^{\top}$ follows the multivariate
latent Gaussian model whose density function is
defined in (\ref{eq:mvt-lg-pois}). One can easily prove the following:
\begin{itemize}
\item[i)] For $j=1,\ldots,p$, the univariate
marginal density for $y_{j}$ is
\[f_{jj}(y_{j}|\mu_{j},\sigma_{jj})=\int h_{j}(y_{j}|x_{j})\phi(x_{j}|\mu_{j},\sigma_{jj})\intd
  x_{j},\]
\item[ii)] For $1\leq s < t \leq p$, the bivariate
marginal density for $y_{s}$ and $y_{t}$ is
\begin{align*}
\begin{split}
&f_{st}(y_{s},y_{t}|\sigma_{st}, \mu_{s},\sigma_{ss},\mu_{t},\sigma_{tt})\\
=&\int h_{s}(y_{s}|x_{1}) h_{t}(y_{t}|x_{2})
\phi\big(\begin{brsm}
   x_{1}\\
   x_{2} 
\end{brsm}\big|\begin{brsm}
    \mu_{s}\\
    \mu_{t} 
\end{brsm},
\begin{brsm}
    \sigma_{ss} & \sigma_{st} \\
    \sigma_{ts}& \sigma_{tt}
\end{brsm}\big)\intd \begin{brsm}
   x_{1}\\
   x_{2} 
\end{brsm}.
\end{split}
\end{align*}
\end{itemize}

There is a rich literature on numerical integration methods for
univariate and bivariate
functions. Hence $f_{jj}$'s and $f_{st}$'s can be efficiently
evaluated. In fact, many composite likelihood methods have been using
these lower-dimensional densities as individual components due to the
fact that they are computationally easier to work
with~\citep{cox2004note}.

Consider the following composite log-likelihood which consists of only
univariate marginal densities:
\begin{align*}
Q_{n}(\bm{\mu},\sigma_{11},\ldots,\sigma_{pp})
= &\sum_{i}^{n}\sum_{j=1}^{p}\ell_{jj}(y_{ij}|\mu_{j},\sigma_{jj}),
\end{align*}
where $\ell_{jj}(y_{j}|\mu_{j},\sigma_{jj}) = \log
f_{jj}(y_{j}|\mu_{j},\sigma_{jj})$ for $j=1,\ldots,p$. 
One can construct the following posterior distribution:
\begin{align*}
\pi_{n}^{*}(\bm{\mu},\sigma_{11},\ldots,\sigma_{pp})\propto &
  \exp\left\{Q_{n}(\bm{\mu},\sigma_{11},\ldots,\sigma_{pp})\right\}\pi(\bm{\mu},
  \sigma_{11},\ldots,\sigma_{pp}),
\end{align*}
given prior 
$\pi(\bm{\mu}, \sigma_{11},\ldots,\sigma_{pp})$.
If we assume prior independence, so that
$\pi(\bm{\mu}, \sigma_{11},\ldots,\sigma_{pp}) =
  \prod_{j=1}^{p}\pi_{jj}(\mu_{j}, \sigma_{jj})$,
and let
\[\pi_{n,jj}^{*}(\mu_{j},
  \sigma_{jj})\propto\exp\left\{\sum_{i}^{n}\ell_{jj}(y_{ij}|\mu_{j},\sigma_{jj})\right\}\pi_{jj}(\mu_{j},
  \sigma_{jj}),\]
it can be shown that
\begin{align} \label{eq:mvtlg-knot}
\pi_{n}^{*}(\bm{\mu},\sigma_{11},\ldots,\sigma_{pp})=
  \prod_{j=1}^{p}\pi_{n,jj}^{*}(\mu_{j},
  \sigma_{jj}).
\end{align}
We have constructed a surrogate posterior distribution for $\bm{\mu}$ and
$\sigma_{jj}$'s. These are the parameters that characterize the univariate
  marginal distributions of the data. The factorized form of the composite likelihood and prior
independence induces posterior independence in $(\mu_{j},
  \sigma_{jj})$'s. Therefore, sampling from
  (\ref{eq:mvtlg-knot}) can be
  split into independently sampling from each $\pi_{n,jj}^{*}(\mu_{j},
  \sigma_{jj})$.

To complete our surrogate posterior distribution, we need some
conditional distribution of $\sigma_{st}$'s given $\bm{\mu}$ and
$\sigma_{jj}$'s. 
Consider the following composite log-likelihood:
\begin{align*}
L_{n}(\bm{\mu},\bm{\Sigma})
= &\sum_{i}^{n}\sum_{s<t}^{p}\ell_{st}(y_{is},y_{it}|\sigma_{st}, \mu_{s},\sigma_{ss},
\mu_{t},\sigma_{tt}),
\end{align*}
where 
$\ell_{st}(y_{s},y_{t}|\sigma_{st}, \mu_{s},\sigma_{ss},\mu_{t},\sigma_{tt}) = \log f_{st}(y_{s},y_{t}|\sigma_{st}, \mu_{s},\sigma_{ss},\mu_{t},\sigma_{tt})$
for $1\leq s < t \leq p$. 
This time we will assume prior conditional
independence, i.e., the prior density takes the following factorized form:
\begin{align*}
\pi(\sigma_{12},\ldots,\sigma_{(p-1)p}|\bm{\mu}, \sigma_{11},\ldots,\sigma_{pp})=&
  \prod_{1\leq s < t \leq p}^{p}\pi_{st}(\sigma_{st}|\mu_{s}, \sigma_{ss}, \mu_{t}, \sigma_{tt}).
\end{align*}
Letting
\begin{align*}
&\pi_{n,st}^{*}(\sigma_{st}|\mu_{s}, \sigma_{ss}, \mu_{t}, \sigma_{tt})\\
\propto&\exp\left\{\sum_{i}^{n}\ell_{st}(y_{is},y_{it}|\sigma_{st},
 \mu_{s}, \sigma_{ss}, \mu_{t}, \sigma_{tt})\right\}\pi_{st}(\sigma_{st}|\mu_{s}, \sigma_{ss}, \mu_{t}, \sigma_{tt}),
\end{align*}
we can construct the following conditional posterior density:
\begin{align} \label{eq:mvtlg-tile}
\begin{split}
\pi_{n}^{*}(\sigma_{12},\ldots,\sigma_{(p-1)p}|\bm{\mu}, \sigma_{11},\ldots,\sigma_{pp})=&\prod_{1\leq s < t \leq
  p}^{p}\pi_{n,st}^{*}(\sigma_{st}|\mu_{s}, \sigma_{ss}, \mu_{t}, \sigma_{tt}).
\end{split}
\end{align}
Similarly, we have posterior conditional independence in
$\sigma_{st}$'s given $\bm{\mu}$ and
$\sigma_{jj}$'s. Therefore, sampling from (\ref{eq:mvtlg-tile}) can be
  split into independently sampling from each
  $\pi_{n,st}^{*}(\sigma_{st}|\mu_{s}, \sigma_{ss}, \mu_{t},
  \sigma_{tt})$.

Combining  (\ref{eq:mvtlg-knot}) and  (\ref{eq:mvtlg-tile}) we
construct the following surrogate posterior density:
\[\pi_{n}^{*}(\bm{\mu},\bm{\Sigma})=\pi_{n}^{*}(\sigma_{12},\ldots,\sigma_{(p-1)p}|\bm{\mu}, \sigma_{11},\ldots,\sigma_{pp}) \pi_{n}^{*}(\bm{\mu},\sigma_{11},\ldots,\sigma_{pp}).\]
To summarize, we have proposed a surrogate posterior distribution which
is a multiplication of component posteriors. These component
posteriors are
based on either univariate or bivariate marginal densities. Sampling from this
posterior can be done via a composite sampling strategy that contains
two steps. In the first step, we sample the parameters that
characterize the univariate marginal densities ($\bm{\mu}$ and
$\sigma_{jj}$'s). In the second step, we plug
the previous samples into the conditional densities and
sample those parameters that characterize the pairwise relationship
($\sigma_{st}$'s). The computation of both steps can be
easily parallelized due to the sparse posterior dependence structure. We term
$\pi_{n}^{*}(\bm{\mu},\bm{\Sigma})$ as a {\em Bayesian mosaic}
posterior under model (\ref{eq:mvt-lg-pois}). A formal definition
will follow.

\subsection{Definition of Bayesian Mosaic}
It can be seen that the 
independence structure in (\ref{eq:mvtlg-knot}) relies on the fact that univariate marginal
distributions do not share parameters, and that the 
conditional independence structure in (\ref{eq:mvtlg-tile}) requires
that the parameters characterizing the pairwise relationships
($\sigma_{st}$'s) only appear in one bivariate marginal distribution.
We term the class of data
distributions that satisfy the above conditions as {\em mosaic-type}. Below is a formal definition.

\begin{dfn}\label{def:mosaic-type}
Suppose $p\in\setnnn$ and $\bm{y}_{1},\ldots,\bm{y}_{n}$
are i.i.d. $p$-dimensional data vectors
from distribution $P_{\bm{\theta}}$ with density function
$f(\bm{y}|\bm{\theta})$. Let $\bm{\theta}_{st}$, $1\leq s\leq t \leq
p$, be non-overlapping sub-vectors of $\bm{\theta}$ such that 
\[\bm{\theta}=\big[\bm{\theta}_{12}^{\top},\ldots, \bm{\theta}_{(p-1)p}^{\top}, \bm{\theta}_{11}^{\top},\ldots, \bm{\theta}_{pp}^{\top}\big]^{\top},\]
then the data distribution $P_{\bm{\theta}}$
is mosaic-type if there exists a collection of density functions
$f_{st}$ for $1\leq s \leq t \leq p$ such that
\begin{itemize}
\item[i)] for $j=1,\ldots,p$, the density of the univariate marginal
  distribution for dimension $j$ is
\[f_{jj}(y_{j}|\bm{\theta}_{jj}).\]
\item[ii)] for $1 \leq t < s \leq p$, the density of the bivariate
  marginal data distribution for dimension $s$ and $t$ is
\[f_{st}(y_{s},y_{t}|\bm{\theta}_{st}, \bm{\theta}_{ss}, \bm{\theta}_{tt}).\]
\end{itemize}
\end{dfn}

We term $\bm{\theta}_{jj}$'s as {\em knots} since they are shared among
multiple bivariate marginal distributions. We term
$\bm{\theta}_{st}$'s as {\em tiles} since they only appear in one
bivariate marginal distribution. In the multivariate latent Gaussian
example,
\[\bm{\theta}_{jj}=\begin{brsm}
    \mu_{j} \\
    \sigma_{jj}
\end{brsm}, \quad \bm{\theta}_{st}=\sigma_{st}.\]
We will show that Definition~\ref{def:mosaic-type} provides a
rich class of models later in \S\ref{ssec:richness}. Although
we require $\bm{y}_{1},\ldots,\bm{y}_{n}$
to be independent for now, to ease our analysis of asymptotic
properties, in practice this requirement can
be relaxed. We provide a more general definition in \S\ref{ssec:generalize}.

Before defining {\em Bayesian mosaic}, we will first introduce some
notation. For $j=1,\ldots,p$, define
\begin{align*}
\ell_{jj}(\bm{\theta}_{jj}, y_{j}) =\log
  f_{jj}(y_{j}|\bm{\theta}_{jj}), \quad Q_{n,j}(\bm{\theta}_{jj})=\sum_{i=1}^{n}\ell_{jj}(\bm{\theta}_{jj}, y_{ij}).
\end{align*}
For $1\leq s<t \leq p$, define
\begin{align*}
&\ell_{st}(\bm{\theta}_{st}, \bm{\theta}_{ss}, \bm{\theta}_{tt},
  y_{s},y_{t}) = \log
  f_{st}(y_{s},y_{t}|\bm{\theta}_{st}, \bm{\theta}_{ss},
                  \bm{\theta}_{tt}),\\
&L_{n,st}(\bm{\theta}_{st}, \bm{\theta}_{ss}, \bm{\theta}_{tt})=\sum_{i=1}^{n}\ell_{st}(\bm{\theta}_{st}, \bm{\theta}_{ss}, \bm{\theta}_{tt},
  y_{is},y_{it}).
\end{align*}
The formal definition of {\em Bayesian mosaic} is given below.

\begin{dfn}\label{def:BM}
Under the setup of Definition~\ref{def:mosaic-type} and considering prior densities
$\pi_{jj}(\bm{\theta}_{jj})$ for $j=1,\ldots,p$ and $\pi_{st}(\bm{\theta}_{st}|\bm{\theta}_{ss},\bm{\theta}_{tt})$ for $1 \leq t < s \leq p$, we
introduce the following:
\begin{itemize}
\item[i)] For $j=1,\ldots,p$, the knot marginal for $\bm{\theta}_{jj}$ is
\begin{align*}
\kappa_{n,j}(\bm{\theta}_{jj}) \propto e^{Q_{n,j}(\bm{\theta}_{jj})}\pi_{jj}(\bm{\theta}_{jj}).
\end{align*}
\item[ii)] For $1\leq t < s \leq p$, the tile conditional for
  $\bm{\theta}_{st}$ given
  $\bm{\theta}_{ss}$ and $\bm{\theta}_{tt}$ is
\begin{align*}
\tau_{n,st}(\bm{\theta}_{st}|\bm{\theta}_{ss},\bm{\theta}_{tt}) \propto e^{L_{n,st}(\bm{\theta}_{st}, \bm{\theta}_{ss}, \bm{\theta}_{tt})}\pi_{st}(\bm{\theta}_{st}|\bm{\theta}_{ss},\bm{\theta}_{tt}).
\end{align*}
\end{itemize}
Then we call
\begin{align}\label{eq:BM}
\tilde{\pi}(\bm{\theta}) = \prod_{j}^{p}\kappa_{n,j}(\bm{\theta}_{jj})\prod_{s<t}^{p}\tau_{n,st}(\bm{\theta}_{st}|\bm{\theta}_{ss},\bm{\theta}_{tt})
\end{align}
a Bayesian mosaic posterior under model $P_{\bm{\theta}}$.
\end{dfn}

\begin{figure}
\centering
\includegraphics[width=.5\linewidth]{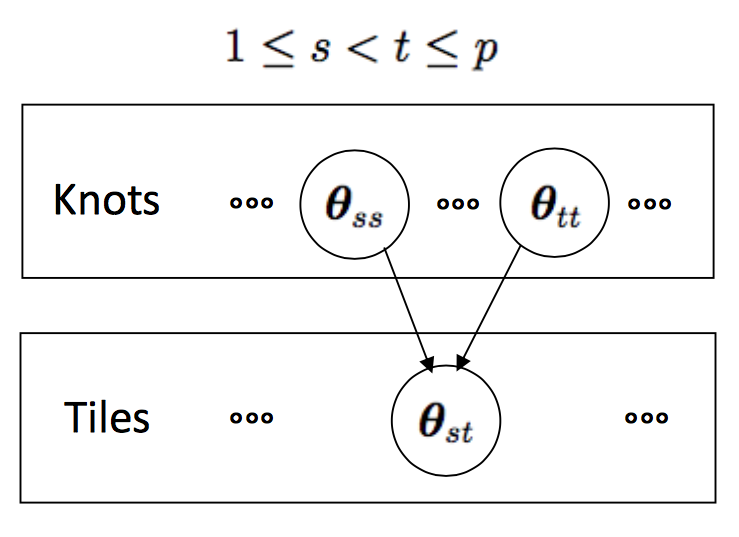}
\caption{DAG representation of a {\em Bayesian mosaic}.}
\label{fig:dag_bayesian_mosaic}
\end{figure}

\subsection{Sampling Bayesian Mosaic} \label{ssec:sampling}
It is easily seen from (\ref{eq:BM}) that the {\em knots} are
marginally independent and the {\em tiles} are conditionally
independent given the {\em knots}. This sparse dependence structure of {\em Bayesian mosaic} 
can be represented by a directed acyclic graph (DAG), as demonstrated
in Figure~\ref{fig:dag_bayesian_mosaic}. Utilizing this structure, we propose a simple parallel sampling strategy
which is summarized in Algorithm~\ref{alg:sampler}, where $M\in\setnnn$
denotes the total number of posterior samples to be collected.

\begin{algorithm}
Step 1
\# on each of the {\em knot marginals} in parallel\\
 \For{$j=1,\ldots,p$}{
     sample $\bm{\theta}_{jj}^{1}, \ldots, \bm{\theta}_{jj}^{M}$ from $\kappa_{n,j}(\bm{\theta}_{jj})$
 }
Step 2
\# on each of the {\em tile conditionals} in parallel\\
\For{$s=2,\ldots,p$}{
  \For{$t=1,\ldots,s-1$}{
    sample $\bm{\theta}_{st}^{m}$ from
    $\tau_{n,st}(\bm{\theta}_{st}|\bm{\theta}_{ss}^{m},\bm{\theta}_{tt}^{m})$,
    for $m=1,\ldots, M$
  }
}
\caption{Parallel sampler for {\em Bayesian mosaic}.}
\label{alg:sampler}
\end{algorithm}

Usually, the {\em knot marginals} and the {\em tile
  conditionals} can not be directly sampled from. We propose to
sample the {\em knot
  marginals} via Metropolis-Hastings (MH) algorithms with the $Q_{n}(\bm{\theta}_{jj})$'s
being evaluated via numerical integration.
Sampling from the {\em tile conditionals} is harder, since the conditional distribution
$\tau_{n,st}(\bm{\theta}_{st}|\bm{\theta}_{ss}^{m},\bm{\theta}_{tt}^{m})$
changes w.r.t. the values of $\bm{\theta}_{ss}^{m}$ and
$\bm{\theta}_{tt}^{m}$. We propose the following three options:
\begin{itemize}
\item[i)] Suppose that the MH sampler on
  $\tau_{n,st}(\bm{\theta}_{st}|\bm{\theta}_{ss}^{m},\bm{\theta}_{tt}^{m})$
  converges rapidly, then for each $m$, one can run the sampler for a
  fixed small number of steps and use the last draw as the sample.
\item[ii)] Suppose that
  $\tau_{n,st}(\bm{\theta}_{st}|\bm{\theta}_{ss}^{m},\bm{\theta}_{tt}^{m})$
  is easy to optimize w.r.t. $\bm{\theta}_{st}$, then one can compute
  the mode and the maximum density value. Then one can either do
  rejection sampling using the maximum density value or obtain the
  Hessian matrix at the mode and approximate the density by its
  Laplace approximation.
\item[iii)] One can directly plug in the posterior means of
the {\em knots} into the {\em tile conditionals} so that they remain
the same across iterations. Simply
substitute
$\tau_{n,st}(\bm{\theta}_{st}|\bm{\theta}_{ss}^{m},\bm{\theta}_{tt}^{m})$
in the second step of Algorithm~\ref{alg:sampler} with
$\tau_{n,st}\left(\bm{\theta}_{st}\big|\frac{1}{M}\sum_{m=1}^{M}\bm{\theta}_{ss}^{m},
  \frac{1}{M}\sum_{m=1}^{M}\bm{\theta}_{tt}^{m}\right)$. 
\end{itemize}

The third option is the fallback plan when the first two are
unavailable. Note that when applying the third option, instead of
sampling from the {\em Bayesian mosaic}, one actually samples from the
following approximation:
\begin{align}\label{eq:plugin-bm}
\prod_{j}^{p}\kappa_{n,j}(\bm{\theta}_{jj})\prod_{s<t}^{p}\tau_{n,st}(\bm{\theta}_{st}|\bm{\theta}^{*}_{ss},\bm{\theta}^{*}_{tt}),
\end{align}
where $\bm{\theta}^{*}_{jj}=\int \kappa_{n,j}(\bm{\theta}_{jj})\intd
\bm{\theta}_{jj}$ is the posterior mean of $\bm{\theta}_{jj}$, for
$j=1,\ldots,p$. 
In \S\ref{ssec:asym-post-mean} we will show that (\ref{eq:plugin-bm})
is still consistent and asymptotically normal in a slightly weaker sense, but will
underestimate the uncertainty compared to the exact {\em Bayesian mosaic}.

\subsection{Handling Parameter Constraints} \label{sec:para-constraint}
In some cases, the model parameters $\bm{\theta}$ live in a constrained
space $\mathcal{T}$. However, the samples from {\em Bayesian mosaic}
do not necessarily also live in this space. For instance, in
(\ref{eq:mvt-lg-pois}), the samples of $\bm{\Sigma}$ from {\em
  Bayesian mosaic} are not guaranteed to be positive
definite. One can easily see this from the fact that the off-diagonal
elements $\sigma_{st}$'s are conditionally independent given the
diagonal elements $\sigma_{jj}$'s. 

To address this, we propose to project the samples from {\em
  Bayesian mosaic} back to the constrained space w.r.t. the Euclidean
distance. Specifically, we solve the
following optimization problem for each sample $\bm{\theta}^{m}$:
\begin{align}
\label{eq:sample-correction}
\tilde{\bm{\theta}}^{m}=\argmin_{\tilde{\bm{\theta}}^{m}\in\mathcal{T}}
  \|\tilde{\bm{\theta}}^{m}-\bm{\theta}^{m}\|.
\end{align}
We term $\tilde{\bm{\theta}}^{m}$'s as the corrected samples from {\em
  Bayesian mosaic}. For many structured constrained
parameter spaces $\mathcal{T}$, (\ref{eq:sample-correction}) has an
analytical solution. For instance, if $\mathcal{T}$ is the cone of
positive definite matrices, then $\tilde{\bm{\theta}}^{m}$ can be
obtained via an eigenvalue decomposition of $\bm{\theta}^{m}$.

In \S\ref{ssec:consistency-normality}, we will prove that the
probability mass of {\em Bayesian mosaic} asymptotically concentrates within a
small neighbourhood of the ``true'' value $\bm{\theta}_{0}$. This
implies that when $n$ is large for many constraints, the majority of the
samples should automatically live inside $\mathcal{T}$ and we only need to
correct the rest. Hence, this correcting step should have minimal impact on the overall
performance.

\subsection{Generalization} \label{ssec:generalize}
In this subsection, we will extend {\em Bayesian mosaic} to dependent
data. We first provide a more general definition of {\em mosaic-type}
data distributions.

\begin{dfn}\label{def:mosaic-type-dependent}
Suppose $p\in\setnnn$ and $\bm{y}_{1},\ldots,\bm{y}_{n}$
are $p$-dimensional data vectors
jointly from distribution $P_{\bm{\theta}}$ with a joint density function
$f(\bm{y}_{1},\ldots,\bm{y}_{n}|\bm{\theta})$. Let $\bm{\theta}_{st}$, $1\leq t\leq s \leq
p$, be non-overlapping sub-vectors of $\bm{\theta}$ such that 
\[\bm{\theta}=\big[\bm{\theta}_{12}^{\top},\ldots, \bm{\theta}_{(p-1)p}^{\top}, \bm{\theta}_{11}^{\top},\ldots, \bm{\theta}_{pp}^{\top}\big]^{\top},\]
then the data distribution $P_{\bm{\theta}}$
is mosaic-type if there exists a collection of density functions
$f_{st}$, $1\leq s \leq t \leq p$ such that
\begin{itemize}
\item[i)] for $j=1,\ldots,p$, the density of the univariate marginal
  distribution for dimension $j$ is
\[f_{jj}(y_{1j},\ldots,y_{nj}|\bm{\theta}_{jj}).\]
\item[ii)] for $1 \leq t < s \leq p$, the density of the bivariate
  marginal data distribution for dimension $s$ and $t$ is
\[f_{st}(y_{1s},\ldots,y_{ns},y_{1t},\ldots,y_{nt}|\bm{\theta}_{st}, \bm{\theta}_{ss}, \bm{\theta}_{tt}).\]
\end{itemize}
\end{dfn}

For $j=1,\ldots,p$, we redefine
$Q_{n,j}(\bm{\theta}_{jj})$ as
\begin{align*}
Q_{n,j}(\bm{\theta}_{jj})=\log f_{jj}(y_{1j},\ldots,y_{nj}|\bm{\theta}_{jj}).
\end{align*}
For $1\leq s<t \leq p$, we redefine
$L_{n,st}(\bm{\theta}_{st}, \bm{\theta}_{ss}, \bm{\theta}_{tt})$ as
\begin{align*}
L_{n,st}(\bm{\theta}_{st}, \bm{\theta}_{ss}, \bm{\theta}_{tt})=\log f_{st}(y_{1s},\ldots,y_{ns},y_{1t},\ldots,y_{nt}|\bm{\theta}_{st}, \bm{\theta}_{ss}, \bm{\theta}_{tt}).
\end{align*}
Then {\em Bayesian mosaic} still follows Definition~\ref{def:BM}.

Under this generalization, one can include random effects and still be able to use {\em Bayesian mosaic}. We will give an
example of such a model in \S\ref{ssec:citation}, where we include
random temporal effects.

\section{Theoretical Analysis} \label{sec:theory}
In this section we will first demonstrate that the {\em mosaic-type}
class contains a rich collection of models. We will then provide 
regularity conditions and prove under these conditions that {\em Bayesian mosaic} is
consistent and asymptotically normal. Moreover, we
will analyze the asymptotic distribution of the {\em tiles}
conditional on the posterior means of the {\em knots}. Finally we will
build a connection between the sampling computational complexity and
the cardinality of the data. We prove the main result and defer other
proofs to the supplement.

\subsection{Richness of the Mosaic-type Distribution Class} \label{ssec:richness}
To evaluate how widely {\em
  Bayesian mosaic} can be applied in practice, it is crucial to
understand how rich the {\em mosaic-type} distribution class is.
The following lemma provides one simple rule to construct new
{\em mosaic-type} distributions from any existing {\em mosaic-type} distributions. With the help of this rule, one can build models
for any type of data with the dependence induced by
latent variables with some underlying {\em mosaic-type} distribution.

\begin{lemma}\label{lem:mosaic-rich}
Suppose that $P_{\bm{\psi}}$ is some data distribution with
density function $f_{0}(\bm{x}|\bm{\psi})$ and consider another data
distribution $P_{\bm{\mu},\bm{\psi}}$ with the following density function:
\begin{align} \label{eq:mosaic-rich}
f(\bm{y}|\bm{\mu},\bm{\psi}) = \int
  f_{0}(\bm{x}|\bm{\psi})\prod_{j=1}^{p}g_{j}(y_{j}|x_{j},\bm{\mu}_{j})\intd \bm{x},
\end{align}
where  $\bm{\mu}=\big(\bm{\mu}_{1}^{\top},\ldots,
\bm{\mu}_{p}^{\top}\big)^{\top}$ and $g_{j}$'s are proper
density functions. If $P_{\bm{\psi}}$ is mosaic-type, so is $P_{\bm{\mu},\bm{\psi}}$.
\end{lemma}

One could choose $P_{\bm{\psi}}$ to be the multivariate Gaussian distribution, and $g_{j}(y_{j}|x_{j},\bm{\mu}_{j})$'s to be any
univariate density. This implies that the {\em mosaic-type}
model class contains the multivariate latent Gaussian models. Note
that besides Gaussian, $P_{\bm{\psi}}$ could also be a Dirichlet or a
multinomial distribution. It is easy to check that both
distributions are {\em mosaic-type}. Moreover, one can construct
arbitrarily complex models by repeatedly
applying Lemma~\ref{lem:mosaic-rich}.

\subsection{Posterior Consistency \& Asymptotic Normality} \label{ssec:consistency-normality}
We start our analysis in a simpler yet
more general setup. Consider $p\in\setnnn$, $d\in\setnnn$ and
$\bm{y}_1,\ldots,\bm{y}_n$ which are i.i.d. $p$-dimensional observations
from distribution $P_{\bm{\theta}}$ possessing a density
$f(\bm{y}|\bm{\theta})$ where $\bm{\theta}\in\mathcal{T}\subset \setr^{d}$. We fix
$\bm{\theta}_0$ to be the ``true value" of the parameters and require that
$\bm{\theta}_0$ is an interior point of $\mathcal{T}$. Consider two non-overlapping
sub-vectors of $\bm{\theta}$, $\bm{\eta}$ and $\bm{\zeta}$, where
$\bm{\eta}$ is $d_{\eta}$-dimensional and $\bm{\zeta}$ is $d_{\zeta}$-dimensional. Let $\bm{\eta}_{0}$
and $\bm{\zeta}_{0}$ be the corresponding ``true values". Consider
two pseudo density functions
$f_{1}(\bm{y}|\bm{\zeta})$ and
$f_{2}(\bm{y}|\bm{\eta},\bm{\zeta})$, which do not have to integrate
to one. In order for proper Bayesian inference, the following
regularity conditions need to hold for both functions. To avoid redundancy, we only define these
conditions for $f_{1}(\bm{y}|\bm{\zeta})$.

\begin{condition} \label{cond:1-1}
The support set $\{\bm{y}: f_{1}(\bm{y}|\bm{\zeta})>0\}$ is the same for all $\bm{\zeta}$.
\end{condition}

\begin{condition} \label{cond:1-2}
Consider $\ell_{1}(\bm{\zeta},
\bm{y})=\log f_{1}(\bm{y}|\bm{\zeta})$. $\ell_{1}(\bm{\zeta},
\bm{y})$ is thrice differentiable with respect to $\bm{\zeta}$ in a
neighborhood $\mathcal{B}_{\bm{\zeta}}(\bm{\zeta}_{0},\delta)$. The
expectations $\mathbb{E}_{\bm{\theta}_{0}}\triangledown_{\bm{\zeta}} \ell_{1}(\bm{\zeta},
\bm{y})$ and $\mathbb{E}_{\bm{\theta}_{0}}\triangledown^{2}_{\bm{\zeta}} \ell_{1}(\bm{\zeta},
\bm{y})$ are both finite and for any {\em multi-index}
$\bm{\alpha}$ for $\bm{\zeta}$ such that $|\bm{\alpha}|=3$, we have
\[\sup_{\bm{\zeta}\in \mathcal{B}_{\bm{\zeta}}(\bm{\zeta}_{0},\delta)} \left|\partial^{\bm{\alpha}}\ell_{1}(\bm{\zeta},
\bm{y})\right| \leq
  M_{\bm{\alpha}}(\bm{y}),\]
and $\mathbb{E}_{\bm{\theta}_{0}}M_{\bm{\alpha}}(\bm{y})<\infty$.
\end{condition}

\begin{condition} \label{cond:1-3}
 Consider $\ell_{1}(\bm{\zeta},
\bm{y})=\log f_{1}(\bm{y}|\bm{\zeta})$. Then
 $\mathbb{E}_{\bm{\theta}_{0}}\triangledown_{\bm{\zeta}} \ell_{1}(\bm{\zeta},
\bm{y}) |_{\bm{\zeta}=\bm{\zeta}_{0}}= \bm{0}$
and 
\begin{align*}
\mathbb{E}_{\bm{\theta}_{0}}\triangledown^{2}_{\bm{\zeta}}\ell_{1}(\bm{\zeta},
\bm{y})
   |_{\bm{\zeta}=\bm{\zeta}_{0}}
= &-\mathbb{E}_{\bm{\theta}_{0}}\big[\triangledown _{\bm{\zeta}} \ell_{1}(\bm{\zeta},
\bm{y})
    \big]\big[\triangledown _{\bm{\zeta}} \ell_{1}(\bm{\zeta},
\bm{y})\big]^{\top}|_{\bm{\zeta}=\bm{\zeta}_{0}}
\end{align*}
Also, the Fisher information
$ -\mathbb{E}_{\bm{\theta}_{0}}\triangledown^{2}_{\bm{\zeta}}\ell_{1}(\bm{\zeta},
\bm{y})
    |_{\bm{\zeta}=\bm{\zeta}_{0}}$
 is positive definite.
\end{condition}

\begin{condition} \label{cond:1-4}
Consider $Q_{n}(\bm{\zeta})=\sum_{i=1}^{n}\log
f_{1}(\bm{y}_{i}|\bm{\zeta})$. For any $\delta>0$, $\exists
\epsilon>0$ such that with
  $P_{\bm{\theta}_{0}}$-probability one
\[\sup_{\bm{\zeta} \notin \mathcal{B}_{\bm{\zeta}}(\bm{\zeta}_{0},\delta)}\frac{1}{n}\big[Q_{n}(\bm{\zeta})
  - Q_{n}(\bm{\zeta}_{0})\big]<-\epsilon\]
for all sufficiently large $n$.
\end{condition}

\begin{condition} \label{cond:1-5}
Consider $Q_{n}(\bm{\zeta})=\sum_{i=1}^{n}\log
f_{1}(\bm{y}_{i}|\bm{\zeta})$ and $\tilde{\bm{\zeta}}_{n}=\argmax_{\bm{\zeta}}Q_{n}(\bm{\zeta})$.
$\tilde{\bm{\zeta}}_{n}$ is consistent at $\bm{\zeta}_{0}$, i.e.,
$\lim_{n\to\infty}\tilde{\bm{\zeta}}_{n}=\bm{\zeta}_{0}$ with $P_{\bm{\theta}_{0}}$-probability one.
\end{condition}

For ease of presentation, we introduce some notation. We define
\begin{align*}
\ell_{1}(\bm{\zeta},
\bm{y}) = \log f_{1}(\bm{y}|\bm{\zeta}), & \quad Q_{n}(\bm{\zeta}) = \sum_{i=1}^{n}\ell_{1}(\bm{\zeta},
\bm{y}_{i}), \\
\ell_{2}(\bm{\eta}, \bm{\zeta},
\bm{y}) = \log f_{2}(\bm{y}|\bm{\eta},\bm{\zeta}), & \quad L_{n}(\bm{\eta},\bm{\zeta}) = \sum_{i=1}^{n}\ell_{2}(\bm{\eta}, \bm{\zeta},
\bm{y}_{i}),\\
\begin{brsm}
    \hat{\bm{\eta}}_{n} \\
    \hat{\bm{\zeta}}_{n}
\end{brsm}
=
  \argmax_{\bm{\eta},\bm{\zeta}}L_{n}(\bm{\eta},\bm{\zeta}), & \quad
  \tilde{\bm{\zeta}}_{n} = \argmax_{\bm{\zeta}}Q_{n}(\bm{\zeta}),
\end{align*}
and
\begin{align*}
\tilde{\bm{I}}_{0}= -\mathbb{E}_{\bm{\theta}_{0}}\triangledown^{2}_{\bm{\zeta}}\ell_{1}(\bm{\zeta},
\bm{y})
    |_{\bm{\theta}=\bm{\theta}_{0}}, & \quad \bm{I}_{0} = -\mathbb{E}_{\bm{\theta}_{0}}\triangledown_{\bm{\eta},\bm{\zeta}}^{2}  \ell_{2}(\bm{\eta}, \bm{\zeta},
\bm{y})
    |_{\bm{\theta}=\bm{\theta}_{0}}.
\end{align*}

Given a prior density $\pi(\bm{\zeta})$, consider the following posterior density of $\bm{\zeta}$:
\[\kappa_{n}\big(\bm{\zeta}\big)\propto
  \exp\{Q_{n}(\bm{\zeta})\}\pi(\bm{\zeta}).\]
Introducing $\bm{\omega}=\sqrt{n}(\bm{\zeta}-\tilde{\bm{\zeta}}_{n})$, 
the posterior density of $\bm{\omega}$ is
\[\pi^{*}_{n,1}(\bm{\omega})=\frac{1}{\sqrt{n}}\kappa_{n}\bigg(\frac{\bm{\omega}}{\sqrt{n}}+\tilde{\bm{\zeta}}_{n}\bigg).\]
The following lemmas state that $\pi^{*}_{n,1}(\bm{\omega})$ is
asymptotically normal under the specified conditions. This lemma is basically the multivariate version of
Theorem 4.2 in \citet{ghosh2007introduction}, hence the proof will be omitted.

\begin{lemma} \label{lem:knot-normality}
Suppose that conditions~\ref{cond:1-1}-\ref{cond:1-5} hold for
$f_{1}(\bm{y}|\bm{\zeta})$, and the
prior density $\pi(\bm{\zeta})$ is continuous and positive
at $\bm{\zeta}_{0}$, then with $P_{\bm{\theta}_{0}}$-probability one
\begin{align}
\lim_{n\to\infty}\int\big|\pi^{*}_{n,1}(\bm{\omega})-\phi\big(\bm{\omega}\big|\bm{0},
  \tilde{\bm{I}}_{0}^{-1}\big)\big|\intd\bm{\omega}= 0. \label{eq:knot-normality}
\end{align}
\end{lemma}

With a prior density $\pi(\bm{\eta}|\bm{\zeta})$, consider the
following conditional
posterior density of $\bm{\eta}$ given $\bm{\zeta}$:
\[\tau_{n}(\bm{\eta}|\bm{\zeta})\propto \exp\big\{L_{n}\big(\bm{\eta}, \bm{\zeta}\big) - L_{n}\big(\hat{\bm{\eta}}_{n}, \hat{\bm{\zeta}}_{n}\big)\big\}\pi(\bm{\eta}|\bm{\zeta}).\]
Introducing $\bm{t}=\sqrt{n}\big(\bm{\eta}-\hat{\bm{\eta}}_{n}\big)$
and $\bm{r}=\sqrt{n}\big(\bm{\zeta}-\hat{\bm{\zeta}}_{n}\big)$,
the conditional posterior density of $\bm{t}$ given $\bm{r}$ can be
written as
\begin{align*}
\pi^{*}_{n,2}(\bm{t}|\bm{r})=&a_{n}^{-1}(\bm{r})\exp\big\{L_{n}\big(\hat{\bm{\eta}}_{n}+\bm{t}/\sqrt{n}, \hat{\bm{\zeta}}_{n}+\bm{r}/\sqrt{n}\big) - L_{n}\big(\hat{\bm{\eta}}_{n}, \hat{\bm{\zeta}}_{n}\big)\big\}
\\
& \qquad \times\pi\big(\hat{\bm{\eta}}_{n}+\bm{t}/\sqrt{n}\big|\hat{\bm{\zeta}}_{n}+\bm{r}/\sqrt{n}\big)
\end{align*}
with $a_{n}(\bm{r})$ being the normalizing constant. We
define
\begin{align*}
\bm{I}^{11}_{0} = -\mathbb{E}_{\bm{\theta}_{0}}\triangledown_{\bm{\eta}}^{2}  \ell_{2}(\bm{\eta}, \bm{\zeta},
\bm{y})
    |_{\bm{\theta}=\bm{\theta}_{0}}, &\quad \bm{I}^{12}_{0} = -\mathbb{E}_{\bm{\theta}_{0}}\triangledown_{\bm{\zeta}} \triangledown_{\bm{\eta}} \ell_{2}(\bm{\eta}, \bm{\zeta},
\bm{y})
    |_{\bm{\theta}=\bm{\theta}_{0}}, \\
\bm{I}^{21}_{0} = -\mathbb{E}_{\bm{\theta}_{0}}\triangledown_{\bm{\eta}} \triangledown_{\bm{\zeta}} \ell_{2}(\bm{\eta}, \bm{\zeta},
\bm{y})
    |_{\bm{\theta}=\bm{\theta}_{0}}, &\quad \bm{I}^{22}_{0} = -\mathbb{E}_{\bm{\theta}_{0}}\triangledown_{\bm{\zeta}}^{2}  \ell_{2}(\bm{\eta}, \bm{\zeta},
\bm{y})
    |_{\bm{\theta}=\bm{\theta}_{0}}.
\end{align*}
It is easily seen that
$\bm{I}_{0}=\begin{brsm}
\bm{I}^{11}_{0} & \bm{I}^{12}_{0} \\
\bm{I}^{21}_{0} & \bm{I}^{22}_{0}
\end{brsm}$.

\begin{theorem} \label{thm:joint-normality}
Suppose that the conditions for Lemma~\ref{lem:knot-normality}
hold, conditions~\ref{cond:1-1}-\ref{cond:1-5} hold for
$f_{2}(\bm{y}|\bm{\eta},\bm{\zeta})$ and that the
prior density $\pi(\bm{\eta}|\bm{\zeta})$ is continuous and positive
at $\begin{brsm}
\bm{\eta}_{0} \\
    \bm{\zeta}_{0}
\end{brsm}$,
then with $P_{\bm{\theta}_{0}}$-probability one
\begin{align}
\lim_{n\to\infty}\int \int \left|\pi_{n}^{*}(\bm{t},\bm{r})-\phi\left(\bm{t}\middle|-\big(\bm{I}^{11}_{0}\big)^{-1}\bm{I}^{12}_{0}\bm{r},\big(\bm{I}^{11}_{0}\big)^{-1}\right) \phi\left(\bm{r}\middle|\bm{\mu}_{n},
  \tilde{\bm{I}}_{0}^{-1}\right)\right|\intd \bm{r} \intd\bm{t}= 0, \label{eq:joint-normality}
\end{align}
where $\pi_{n}^{*}(\bm{t},\bm{r})=\pi_{n,2}^{*}(\bm{t}|\bm{r}) \pi_{n,1}^{*}(\bm{r})$ is the joint
posterior density of $\bm{t}$ and $\bm{r}$ and 
$\bm{\mu}_{n} =
\sqrt{n}\big(\tilde{\bm{\zeta}}_{n}-\hat{\bm{\zeta}}_{n}\big)$.
\end{theorem}

\begin{corollary} \label{coro:joint-consistency}
Under the same setup in Theorem~\ref{thm:joint-normality},
$\tau_{n}\left(\bm{\eta}|\bm{\zeta}\right) \kappa_{n}\left(\bm{\zeta}\right)$ is consistent at
$\bm{\eta}_{0}$ and
$\bm{\zeta}_{0}$.
\end{corollary}

We can directly apply Theorem~\ref{thm:joint-normality} and
Corollary~\ref{coro:joint-consistency} to  analyze the asymptotic
properties of {\em Bayesian mosaic}. All we need is to let
$f_{1}(\bm{y}|\bm{\zeta})$ be the multiplication of the densities of
the univariate marginal distributions,
\begin{align*}
& f_{1}(\bm{y}|\bm{\zeta}) =
                           \prod_{j=1}^{p}f_{jj}(y_{j}|\bm{\theta}_{jj}),
\end{align*}
$f_{2}(\bm{y}|\bm{\theta},\bm{\zeta})$ be the multiplication of the densities of
the bivariate marginal distributions,
\begin{align*}
&f_{2}(\bm{y}|\bm{\theta},\bm{\zeta}) = \prod_{1\leq s<t\leq
                                       p}f_{st}(y_{s},y_{t}|\bm{\theta}_{st},
                                       \bm{\theta}_{ss},
                                       \bm{\theta}_{tt}),
\end{align*}
$\pi(\bm{\zeta})$ be the multiplication of the prior densities for
{\em knots},
\begin{align*}
&\pi(\bm{\zeta}) =
                           \prod_{j=1}^{p}\pi_{jj}(\bm{\theta}_{jj}),
\end{align*}
and $\pi(\bm{\eta}|\bm{\zeta})$ be the multiplication of the
conditional prior densities for {\em tiles},
\begin{align*}
&\pi(\bm{\eta}|\bm{\zeta}) = \prod_{1\leq s<t\leq
                                       p}\pi_{st}(\bm{\theta}_{st}|\bm{\theta}_{ss},\bm{\theta}_{tt}).
\end{align*}

We immediately have
\begin{align*}
\ell_{1}(\bm{\zeta},
\bm{y}) = \sum_{j=1}^{p}\ell_{jj}(y_{j}|\bm{\theta}_{jj}), & \quad \ell_{2}(\bm{\eta}, \bm{\zeta},
\bm{y}) = \sum_{1\leq s < t \leq p} \ell_{st}(y_{s},y_{t}|\bm{\theta}_{st}, \bm{\theta}_{ss}, \bm{\theta}_{tt}).
\end{align*}

It can be shown that
\begin{align*}
\kappa_{n}\big(\bm{\zeta}\big)\propto &\exp\bigg\{\sum_{i=1}^{n}\ell_{1}(\bm{\zeta},
\bm{y}_{i})\bigg\}\pi(\bm{\zeta}) \\
= &
    \exp\bigg\{\sum_{j=1}^{p}\sum_{i=1}^{n}\ell_{jj}(y_{ij}|\bm{\theta}_{jj})\bigg\}\prod_{j=1}^{p}\pi_{jj}(\bm{\theta}_{jj})\\
= &
    \prod_{j=1}^{p}\exp\bigg\{\sum_{i=1}^{n}\ell_{jj}(y_{ij}|\bm{\theta}_{jj})\bigg\}\pi_{jj}(\bm{\theta}_{jj})\\
= & \prod_{j=1}^{p}\kappa_{n,jj}\big(\bm{\theta}_{jj}\big).
\end{align*}
Similarly we can show that
\[\tau_{n}(\bm{\eta}|\bm{\zeta}) =
  \prod_{s<t}^{p}\tau_{n,st}(\bm{\theta}_{st}|\bm{\theta}_{ss},
  \bm{\theta}_{tt}).\]
Therefore $\tau_{n}(\bm{\eta}|\bm{\zeta})
\kappa_{n}\big(\bm{\zeta}\big)$ is exactly the {\em Bayesian mosaic}, specifically,
\[\tau_{n}(\bm{\eta}|\bm{\zeta}) \kappa_{n}\big(\bm{\zeta}\big) =
  \prod_{s<t}^{p}\tau_{n,st}(\bm{\theta}_{st}|\bm{\theta}_{ss},
  \bm{\theta}_{tt})\prod_{j=1}^{p}\kappa_{n,jj}\big(\bm{\theta}_{jj}\big)
  = \tilde{\pi}_{n}(\bm{\theta}),\]
where $\bm{\theta}=\begin{brsm}
\bm{\eta} \\
    \bm{\zeta}
\end{brsm}$. Consequently, Theorem~\ref{thm:joint-normality} and
Corollary~\ref{coro:joint-consistency} can be used directly
to analyze the asymptotic properties of {\em Bayesian mosaic}. The
following lemma
provides sufficient conditions for the regularity
conditions for Theorem~\ref{thm:joint-normality} to hold. The proof of
this lemma is straightforward and hence is omitted.
\begin{lemma} \label{lem:bm-condition}
Suppose that for $j=1,\ldots,p$, $f_{jj}(\bm{y}|\bm{\eta}_{jj})  =
f_{jj}(y_{j}|\bm{\eta}_{jj})$ satisfies
conditions~\ref{cond:1-1}-\ref{cond:1-5}, and
for $1\leq s < t \leq p$, $f_{st}(\bm{y}|\bm{\eta}_{st}, \bm{\zeta}_{st})  =
f_{st}(\bm{y}|\bm{\eta}_{st}, \bm{\eta}_{ss}, \bm{\eta}_{t})$ satisfies
conditions~\ref{cond:1-1}-\ref{cond:1-5}.
Then conditions~\ref{cond:1-1}-\ref{cond:1-5} also hold for
both $f_{1}(\bm{y}|\bm{\zeta})$ and $f_{2}(\bm{y}|\bm{\eta},\bm{\zeta})$.
\end{lemma}

Recall that $\bm{t}=\sqrt{n}\big(\bm{\eta}-\hat{\bm{\eta}}_{n}\big)$
and $\bm{r}=\sqrt{n}\big(\bm{\zeta}-\hat{\bm{\zeta}}_{n}\big)$, then
the {\em Bayesian mosaic} of $\bm{t}$ and $\bm{r}$ can be written as
$\tilde{\pi}_{n}^{*}(\bm{t},\bm{r})=\frac{1}{\sqrt{n}}\tilde{\pi}_{n}(\hat{\bm{\eta}}_{n}+\bm{t}/\sqrt{n},\hat{\bm{\zeta}}_{n}+\bm{r}/\sqrt{n})$. Applying
Lemma~\ref{lem:bm-condition} and Theorem~\ref{thm:joint-normality}, we
know that
if the requirements of Lemma~\ref{lem:bm-condition} are satisfied, with $P_{\bm{\theta}_{0}}$-probability one
\begin{align*}
\lim_{n\to\infty}\int \int \left|\pi_{n}^{*}(\bm{t},\bm{r})-\phi\left(\bm{t}\middle|-\big(\bm{I}^{11}_{0}\big)^{-1}\bm{I}^{12}_{0}\bm{r},\big(\bm{I}^{11}_{0}\big)^{-1}\right) \phi\left(\bm{r}\middle|\bm{\mu}_{n},
  \tilde{\bm{I}}_{0}^{-1}\right)\right|\intd \bm{r} \intd \bm{t}= 0,
\end{align*}
where $\tilde{\bm{I}}_{0} = -\mathbb{E}_{\bm{\theta}_{0}}\triangledown^{2}_{\bm{\zeta}}\ell_{1}(\bm{\zeta},
\bm{y})|_{\bm{\theta}=\bm{\theta}_{0}}$, $\bm{I}^{11}_{0} = -\mathbb{E}_{\bm{\theta}_{0}}\triangledown_{\bm{\eta}}^{2}  \ell_{2}(\bm{\eta}, \bm{\zeta},
\bm{y})|_{\bm{\theta}=\bm{\theta}_{0}}$ and $\bm{I}^{12}_{0} = -\mathbb{E}_{\bm{\theta}_{0}}\triangledown_{\bm{\zeta}} \triangledown_{\bm{\eta}} \ell_{2}(\bm{\eta}, \bm{\zeta},
\bm{y})|_{\bm{\theta}=\bm{\theta}_{0}}$.

To gain more insight on what the asymptotic covariance of
$\tilde{\pi}_{n}^{*}(\bm{t},\bm{r})$ is, we will look at
$\bm{I}^{11}_{0}$, $\bm{I}^{12}_{0}$ and $\tilde{\bm{I}}_{0}$ in more
detail. For
$j=1,\ldots,p$, we define
\begin{align*}
\bm{\Sigma}_{jj} = \mathbb{E}_{\bm{\theta}_{0}}\triangledown_{\bm{\theta}_{jj}}^{2}
\ell_{jj}(\bm{\theta}_{jj}, y_{j})
    |_{\bm{\theta}=\bm{\theta}_{0}}.
\end{align*} 
Since $\ell_{1}(\bm{\zeta},
\bm{y}) = \sum_{j=1}^{p}\ell_{jj}(y_{j}|\bm{\theta}_{jj})$ and
$\bm{\zeta}=\left(\bm{\theta}_{11},\ldots,\bm{\theta}_{pp}\right)^{\top}$,
it is easy to see that
\[\tilde{\bm{I}}_{0}=
\begin{bmatrix}
\bm{\Sigma}_{11} & \bm{0} & \cdots & \bm{0}\\
\bm{0} & \bm{\Sigma}_{22} & \cdots & \bm{0}\\
\vdots& \vdots & \ddots & \vdots\\
\bm{0} &\bm{0} & \cdots & \bm{\Sigma}_{pp}
\end{bmatrix}.
\]
For
$1\leq s < t \leq p$, we define
\begin{align*}
\bm{\Sigma}_{st} = \mathbb{E}_{\bm{\theta}_{0}}\triangledown_{\bm{\theta}_{st}}^{2}
\ell_{st}(\bm{\theta}_{st}, \bm{\theta}_{ss}, \bm{\theta}_{tt}, y_{s},y_{t})
    |_{\bm{\theta}=\bm{\theta}_{0}}.
\end{align*}
Since $\ell_{2}(\bm{\eta}, \bm{\zeta},
\bm{y}) = \sum_{1\leq s < t \leq p}
\ell_{st}(y_{s},y_{t}|\bm{\theta}_{st}, \bm{\theta}_{ss},
\bm{\theta}_{tt})$ and
$\bm{\eta}=\left[\bm{\theta}_{12},\ldots,\bm{\theta}_{(p-1)p}\right]^{\top}$, it is easy to see that
\[\bm{I}^{11}_{0} =
\begin{bmatrix}
\bm{\Sigma}_{12} & \bm{0} & \cdots & \bm{0}\\
\bm{0} & \bm{\Sigma}_{13} & \cdots & \bm{0}\\
\vdots& \vdots & \ddots & \vdots\\
\bm{0} &\bm{0} & \cdots & \bm{\Sigma}_{(p-1)p}
\end{bmatrix}.
\]
Note that the $\bm{\Sigma}_{st}$'s are ordered in a row-major manner
on the diagonal of $\bm{\Sigma}_{0}$. For
$1\leq s < t \leq p$ and $j=1,\ldots,p$, we define
\begin{align*}
\bm{\Sigma}_{st,j} = \mathbb{E}_{\bm{\theta}_{0}}\triangledown_{\bm{\theta}_{jj}}\triangledown_{\bm{\theta}_{st}}
\ell_{st}(\bm{\theta}_{st}, \bm{\theta}_{ss}, \bm{\theta}_{tt}, y_{s},y_{t})
    |_{\bm{\theta}=\bm{\theta}_{0}}.
\end{align*}
It can be shown that
\[\bm{I}^{12}_{0}=
\begin{bmatrix}
\bm{\Sigma}_{12,1} & \bm{\Sigma}_{12,2} & \cdots & \bm{\Sigma}_{12,p}\\
\bm{\Sigma}_{13,1} & \bm{\Sigma}_{13,2} & \cdots &  \bm{\Sigma}_{13,p}\\
\vdots& \vdots & \ddots & \vdots\\
\bm{\Sigma}_{(p-1)p,1} &\bm{\Sigma}_{(p-1)p,2}  & \cdots &\bm{\Sigma}_{(p-1)p,p}
\end{bmatrix},
\]
where $\bm{\Sigma}_{st,j}$'s are
ordered in a row-major manner within their column for
$j=1\ldots,p$. Note that $\bm{I}^{12}_{0}$ is sparse since
$\bm{\Sigma}_{st,j}=\bm{0}$ if $j\neq s$ and $j\neq t$.

$\tilde{\bm{I}}_{0}^{-1}$ is the marginal variance for
$\bm{r}$. It is block diagonal due to the posterior independence of the
{\em knots}. It can be seen that each block is the Fisher information
induced by a univariate marginal data distribution. $\left(\bm{I}^{11}_{0}\right)^{-1}$
is the conditional variance for $\bm{t}$. It is also block
diagonal due to the conditional independence of the {\em tiles} given
the {\em knots}. Each block is the Fisher information
induced by a bivariate marginal data distribution. $\bm{I}^{12}_{0}$
characterizes the connection between {\em knots} and {\em tiles}. Its
sparsity is due to the fact that $\bm{\theta}_{st}$ given $\bm{\theta}_{ss}$
and $\bm{\theta}_{tt}$ is conditionally
independent of other {\em knots}.

\subsection{Asymptotic Properties of the Posterior
  Mean} \label{ssec:asym-post-mean}
Under the same setup of Lemma~\ref{lem:knot-normality}, define $\bm{\zeta}_{n}^{*}$
as the posterior mean w.r.t. $\kappa_{n}\big(\bm{\zeta}\big)$, i.e.,
$\bm{\zeta}_{n}^{*} = \int
  \bm{\zeta}\kappa_{n}\big(\bm{\zeta}\big)\intd \bm{\zeta}$.
We can prove the following lemma.
\begin{lemma} \label{lem:knot-posterior-mean}
Suppose the conditions for Lemma~\ref{lem:knot-normality} hold and
that the prior
$\pi(\bm{\zeta})$ has a finite expectation, then
$\lim_{n\to\infty}\sqrt{n}\big(\bm{\zeta}_{n}^{*}-\tilde{\bm{\zeta}}_{n}\big)= 0$
with $P_{\bm{\theta}_{0}}$-probability one.
\end{lemma}
Lemma~\ref{lem:knot-posterior-mean} is a multivariate version of
Theorem 4.3 in \citet{ghosh2007introduction}. It states that the
posterior mean is approximately the same as the MLE when $n$ is large.

Now we will investigate sampling from {\em tile conditionals} by
directly plugging in the posterior mean of the {\em knots}. Under the
same setup of Theorem~\ref{thm:joint-normality}, if we plug
$\bm{\zeta}_{n}^{*}$ into the conditional density
$\tau_{n}(\bm{\eta}|\bm{\zeta})$, we will get the following posterior distribution
\[\tau_{n}(\bm{\eta}|\bm{\zeta}_{n}^{*})
  \kappa_{n}\big(\bm{\zeta}\big),\]
which is different from the exact {\em Bayesian mosaic}
posterior. Recalling that
$\bm{t}=\sqrt{n}\big(\bm{\eta}-\hat{\bm{\eta}}_{n}\big)$ and letting $\pi_{n,3}^{*}(\bm{t})=\frac{1}{\sqrt{n}}\tau_{n}(\hat{\bm{\eta}}_{n}+\bm{t}/\sqrt{n}|\bm{\zeta}_{n}^{*})$, the following theorem states that
$\pi_{n,3}^{*}(\bm{t})$ is also asymptotic normal in a slightly weaker sense.
\begin{theorem} \label{thm:plugin-tile-normality}
Suppose that the conditions for Theorem~\ref{thm:joint-normality} hold
that the prior
$\pi(\bm{\zeta})$ has a finite expectation , then 
\begin{align}
\int \left|
  \pi_{n,3}^{*}(\bm{t})-\phi\left(\bm{t}\middle|-\big(\bm{I}^{11}_{0}\big)^{-1}\bm{I}^{12}_{0}\bm{\mu}_{n},\big(\bm{I}^{11}_{0}\big)^{-1}\right)
\right|d \bm{t}\overset{P_{\bm{\theta}_{0}}}{\to} 0. \label{eq:tile-normality}
\end{align}
\end{theorem}
Note that the integral in (\ref{eq:tile-normality}) converges to zero
in probability, which is slightly weaker than the almost surely convergence in
Theorem~\ref{thm:joint-normality}. Moreover,
Theorem~\ref{thm:joint-normality} implies that
\[\lim_{n\to\infty}\int \left|\int \pi_{n}^{*}(\bm{t},\bm{r})\intd
  \bm{r}-\phi\left(\bm{t}\middle|-\big(\bm{I}^{11}_{0}\big)^{-1}\bm{I}^{12}_{0}\bm{\mu}_{n},\big(\bm{I}^{11}_{0}\big)^{-1}+\bm{\Lambda}_{0}\right)\right|\intd
  \bm{t}=0,\]
where
\[\bm{\Lambda}_{0} =
  \bm{I}^{12}_{0}\big(\bm{I}^{11}_{0}\big)^{-1}\tilde{\bm{I}}_{0}\big(\bm{I}^{11}_{0}\big)^{-1}
  \bm{I}^{21}_{0}\]
is positive semi-definite. This indicates that plugging in the
posterior mean leads to some under-estimation of uncertainty as expected.

\subsection{Asymptotic Bound on Computational Complexity} \label{ssec:ab-complexity}
We finish this section by investigating the per-iteration computational complexity of
sampling from {\em Bayesian mosaic} when the data are discrete.
We start by analyzing sampling from the {\em knot marginal}s. Recall
\begin{align} \label{eq:ccb-kappa}
\kappa_{j}(\bm{\theta}_{jj}) \propto e^{\sum_{i=1}^{n}\ell_{jj}(\bm{\theta}_{jj}, y_{ij})}\pi_{jj}(\bm{\theta}_{jj}).
\end{align}
For discrete data, we assume the cardinality of
$\left\{y_{1j},\ldots,y_{nj}\right\}$ is $K\in\setnnn$ and that
$y_{i_{1}j},\ldots,y_{i_{K}j}$ are $K$ unique values of $y_{1j},\ldots,y_{nj}$. For
$k=1,\ldots,K$, we define
$n_{k}=\sum_{i=1}^{n}\mathbbm{1}\{y_{ij}=y_{i_kj}\}$. Then (\ref{eq:ccb-kappa}) can be written as
\begin{align*}
\kappa_{j}(\bm{\theta}_{jj}) \propto e^{\sum_{k=1}^{K}n_{k}\ell_{jj}(\bm{\theta}_{jj}, y_{i_{k}j})}\pi_{jj}(\bm{\theta}_{jj}).
\end{align*}
Clearly the per-iteration computational complexity of sampling from
$\kappa_{n,j}(\bm{\theta}_{jj})$ is dominated by evaluating
$\sum_{k=1}^{K}n_{k}\ell_{jj}(\bm{\theta}_{jj}, y_{i_{k}j})$, which
scales linearly with $K$. It is easily seen that $K$ is bounded by $\max_{1\leq i \leq n} y_{ij} -
\min_{1\leq i \leq n} y_{ij}$. For simplicity, we assume that the data
only take positive values, which implies that $K$ is upper-bounded by
$\max_{1\leq i \leq n} y_{ij}$, whose asymptotic distribution is
studied in extreme value theory. Since this asymptotic distribution is model
specific, we use the
rounded multivariate Gaussian model of
\citet{canale2011bayesian} as an illustration.
This model is a special case of the multivariate
latent Gaussian model defined in
(\ref{eq:mvt-lg-pois}) with
$h_{j}(y_{j}|x_{j})=\mathbbm{1}\{x_{j}>0\}\ceil{x_{j}}$. Basically
$h_{j}(y_{j}|x_{j})$ is a rounding function that rounds $x_{j}$ to the smallest integer larger than it while
mapping all $x_{j}$ below 0 to 0.

\begin{lemma} \label{coro:knot-complexity-bound}
Consider model (\ref{eq:mvt-lg-pois}) with
$h_{j}(y_{j}|x_{j})=\mathbbm{1}\{x_{j}>0\}\ceil{x_{j}}$, for
$j=1,\ldots,p$, we have that $\forall \delta>0$, $\exists N\in\setnnn$
such that $\forall n>N$,
\begin{align} \label{eq:coro-stdn-extreme}
pr\left[\max_{1\leq i\leq
  n}y_{ij}<\sqrt{\sigma_{jj}}\left(\frac{2\delta}{\sqrt{\log n}}+\sqrt{2\log n}\right)+\mu_{j}+1\right]> e^{-\exp(-\delta/2)}.
\end{align}
\end{lemma}
Intuitively,
(\ref{eq:coro-stdn-extreme}) implies that $K$ is at most
$O(\sqrt{\log n})$ with high probability. Similarly, one can show that
the computational
complexity of evaluating the data likelihood of any bivariate marginal
distribution is at most $O(\log n)$. 

To summarize, we have shown that the per-iteration computational
complexity is linear in the cardinality of the discrete
observations. This cardinality can be bounded by the data maxima; hence its
asymptotic distribution can
be analyzed using standard extreme value theory. We have shown that
the per-iteration complexity is at most $O(\log n)$ with high
probability for the rounded multivariate Gaussian model.

\section{Experiments} \label{sec:experiment}
The performance of {\em Bayesian Mosaic} will be illustrated via two
simulation studies and a citation network application.
The first simulation study demonstrates the superiority of {\em Bayesian
  Mosaic} over DA-MCMC for imbalanced count data. The second
simulation study demonstrates that {\em Bayesian
  Mosaic} achieves similar accuracy with a provably more scalable
computational complexity for large balanced count data. {\em Bayesian mosaic} is also applied to a citation count
dataset to infer the overlapping structure of a group of researchers' interests.

All experiments are conducted in R on a machine with 12 3.50 GHz Intel(R) Xeon(R)
CPU E5-1650 v3 processors. All results are based on 100 replicate experiments.

\subsection{Multivariate log-Gaussian Mixture of Poisson} \label{ssec:sim-study-poi}
In the first simulation study, we considered a special case of multivariate
latent Gaussian models with $h_{j}(y_{j}|x_{j})$ being the density
function of a Poisson distribution whose rate parameter equals $e^{x_{j}}$.
We generated 100 datasets for each unique data dimensionality $p$ in
$\{3, 5, 7\}$. We fixed the sample size to be 10000. For each synthetic dataset we randomly generated
$\bm{\mu}$ and $\bm{\Sigma}$ from some distribution so that the
simulated dataset has an excessive amount of zeros. More specifically, for $j=1,\ldots,p$, we generated
$\mu_{j}$ from $\mbox{Unif}(-4,-3)$ and $\sigma_{jj}$ from
$\mbox{Unif}(0.5,1)$. We randomly generated a correlation matrix from
the standard LKJ distribution \cite{lewandowski2009generating} and
then combined this correlation matrix with
$\sigma_{11},\ldots,\sigma_{pp}$ into $\bm{\Sigma}$. Roughly
90\% of the simulated data entries are zeros.

We used weakly-informative priors in both simulation studies. Specifically, for $j=1,\ldots,p$,
\begin{align*}
\pi_{jj}(\mu_{j},\sigma_{jj}) \propto \sigma_{jj}^{-1/2}\mathbbm{1}\left\{\left|\mu_{j}\right|<A, 0<\sigma_{jj}<B\right\},
\end{align*}
where $A>0$ and $B>0$. For $1\leq s<t\leq p$,
\begin{align*}
\pi_{jj}(\sigma_{st}|\mu_{s},\sigma_{ss},\mu_{t},\sigma_{tt}) \propto \mathbbm{1}\left\{\left|\sigma_{st}\right|<\sqrt{\sigma_{ss}\sigma_{tt}}\right\}.
\end{align*}
The propriety of the posterior is guaranteed since the
support of the prior is compact. Moreover, for sufficiently large $A$
and $B$, the posterior becomes
insensitive to the choice of $A$ and $B$~\citep{gelman2006prior}. We
let $A=100$ and $B=10$. We used a similar prior specification in
citation count application.

\begin{table}
\centering
\caption{Mean Square Error Comparison.\protect\footnotemark }\label{tab:pois-mse-highlight}
\begin{tabular}{ l | c | c | c | c}
\hline
   \multicolumn{2}{l|}{} &$p=3$& $p=5$ & $p=7$
  \\ \cline{1-5}
   \multirow{4}{*}{{\em Bayesian Mosaic}} & $\rho$ & $6.79$ ($9.76$)
                                                   & $5.78$ ($8.14$) & $5.62$ ($7.92$)\\ \hhline{~|-|-|-|-|}
    &$s$ & $5.9$ ($9.54$) & $5.86$ ($9.13$) & $5.86$ ($8.63$)\\ \hhline{~|-|-|-|-|}
    &$\mu$& $1.74$ ($2.39$) & $1.74$ ($2.71$) & $1.7$ ($2.52$) \\\hline
\multirow{3}{*}{DA-MCMC} & $\rho$ & $8.41$ ($15.3$) & $9.93$ ($14.8$)
                                          & $10$ ($14.4$) \\\hhline{~|-|-|-|-|}
    &$s$ & $150$  ($1892$) & $303$ ($3020$) & $443$ ($3744$) \\ \hhline{~|-|-|-|-|}
    &$\mu$ & $54.9$ ($345$) & $123$ ($522$) & $126$ ($571$) \\\hline
  \end{tabular}
\end{table}

Normal independent MH sampler was implemented for sampling the {\em
knot marginals} for 200 iterations with the
first 100 as burn-in. We then
approximated {\em tile conditionals} via
Laplace approximation and drew 100 samples of the {\em tiles} from the resulting conditional
Gaussian distribution given the previous draws of the {\em knots}.
On average, a single run with the computation distributed to 11
parallel workers took 90 seconds for $p=3$, 126 seconds for $p=5$ and
227 seconds for $p=7$.
As a comparison,
we ran DA-MCMC sampler for the 5 times the amount of
time with the computation tasks within each iteration distributed to
11 parallel workers as well. In both simulation studies, we
gave DA-MCMC an unfair advantage by 
initializing the parameter values at the true values.
\footnotetext{All numbers have been multiplied by 100.}

We first compared 
accuracies of estimating the model parameters w.r.t. square error
loss.
Average MSE within each group are presented in
Table~\ref{tab:pois-mse-highlight}, where the number in the
parenthesis is the standard error. It can be clearly seen that the
estimates based on {\em Bayesian mosaic} outperforms
those based on DA-MCMC samples in terms of square error loss.

\begin{table}
\centering
\caption{Empirical Coverage Comparison}\label{tab:pois-converage-highlight}
\begin{tabular}{ l | c | c | c | c}
\hline
   \multicolumn{2}{l|}{} &$p=3$& $p=5$ & $p=7$
  \\ \cline{1-5}
   \multirow{3}{*}{{\em Bayesian Mosaic}} & $\rho$
                               & 95\% & 93.9\% & 93\%\\ \hhline{~|-|-|-|-|}
    &$s$ & 93.7\% & 93.4\% & 94.1\% \\ \hhline{~|-|-|-|-|}
    &$\mu$& 93\% & 92.6\% & 93.7\%\\\hline
\multirow{3}{*}{DA-MCMC} & $\rho$ & 64\% & 56.8\% & 59.3\% \\\hhline{~|-|-|-|-|}
    &$s$ & 41.7\% & 32.4\% & 31.1\% \\ \hhline{~|-|-|-|-|}
    &$\mu$ & 37\% & 24.8\% & 24.9\% \\\hline
  \end{tabular}
\end{table}

We evaluated {\em Bayesian mosaic}'s performance in quantifying
the uncertainty through the empirical coverage (EC) of credible intervals. Average EC's are
presented in Table~\ref{tab:pois-converage-highlight}. The empirical
coverages of {\em Bayesian mosaic} are close to 95\%, indicating good
uncertainty quantification, whereas the empirical coverages based on
DA-MCMC are terribly off.

\subsection{Rounded Multivariate Gaussian}
In the second simulation study, we considered 
the rounded multivariate Gaussian model~\citep{canale2011bayesian}
given in \S\ref{ssec:ab-complexity}. We fixed the sample size to be 10000, data dimensionality $p=4$ and generated 100 datasets. For each synthetic dataset we randomly generated
$\bm{\mu}$ and $\bm{\Sigma}$ from some distribution so that the
simulated data are well balanced (majority of the data entries are non-zero). More
specifically, for $j=1,\ldots,p$, we generated
$\mu_{j}$ from $\mbox{Unif}(4,5)$ and $\sigma_{jj}$ from
$\mbox{Unif}(1,1.5)$. The analysis
was done exactly as in \S\ref{ssec:sim-study-poi}.

The average MSE and EC are
summarized in Table~\ref{tab:rn-performance-highlight}. DA-MCMC seems to do slightly better than
{\em Bayesian mosaic}. But the difference in performance is
marginal. Due to the limitation of computation power for DA-MCMC, we did not do
experiments with larger sample size $n$. According to our discussion in
\S\ref{ssec:ab-complexity}, the per-iteration computational complexity of {\em
  Bayesian mosaic} is roughly $O(\log n)$ while that of DA-MCMC is
$O(n)$. This implies that {\em Bayesian mosaic} should be favored in
large sample size applications even if data are well balanced.

\begin{table}
\centering
\caption{Performance Comparison}\label{tab:rn-performance-highlight}
\begin{tabular}{ l | c | c | c }
\hline
   \multicolumn{2}{l|}{} & MSE\protect\footnotemark & EC\\ \cline{1-4}
  \multirow{3}{*}{{\em Bayesian Mosaic}} & $\rho$ & $1.07$ ($1.59$) &90.2\% \\ \hhline{~|-|-|-|}
    &$s$ & $3.23$ ($4.33$) & 94\% \\ \hhline{~|-|-|-|}
    &$\mu$& $1.45$ ($1.89$) & 91.5\% \\\hline
\multirow{3}{*}{DA-MCMC} & $\rho$ & $0.91$ ($1.37$) & 94.2\% \\\hhline{~|-|-|-|}
    &$s$ & $3.29$ ($4.34$) & 95.3\% \\ \hhline{~|-|-|-|}
    &$\mu$ & $1.45$ ($1.91$) & 91.3\% \\\hline
  \end{tabular}
\end{table}

\footnotetext{All numbers have been multiplied by $10^{4}$.}

\subsection{Citation Network Application} \label{ssec:citation}
In this study,
we considered a real-world citation network
dataset~\citep{tang2008arnetminer} that contains
  papers and citation relationships from a computer
  science bibliography website called DBLP.
Our goal is to study the overlapping structure of a group of
researchers' interests. Intuitively, two researchers who have
many research interests in common tend to be cited together more
frequently. Meanwhile, we also want to see how is the research impact
of these researchers varying in time. We hand-picked 11 active
researchers\footnote{Michael Jordan, Robert Brunner, Yann LeCun,
  Andrew McCallum, Chih-Jen Lin, Christopher Bishop, Yoshua Bengio,
  David Blei, Padhraic Smyth, Richard Sutton, Guillermo Sapiro} 
in the machine learning community.

\begin{figure} 
\centering
     \includegraphics[width=0.6\textwidth]{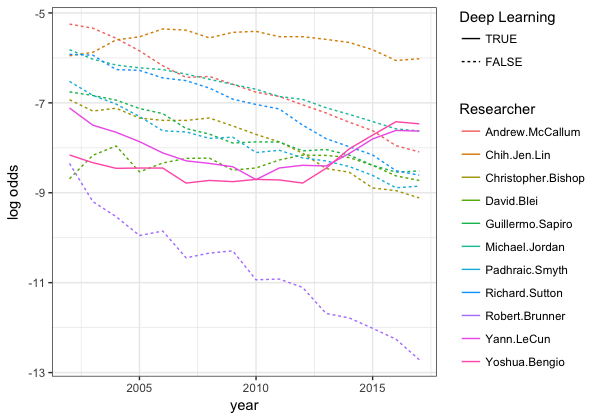}
\caption{Visualizing the posterior mean of $\bm{\mu}_{t}$'s by
  researchers.}
\label{fig:mut}
\end{figure}

In processing the database, we focused on the machine learning
literature and removed irrelevant papers. When counting the number of
citations, we ignored papers co-authored by multiple researchers
in our hand-picked group. The final dataset contains roughly $80000$
$11$-dimensional observations with each one being the number of
citations of a certain paper go to each of the 11 researchers. We used $i$ as the index for papers and
 $j$ as the index for the researchers. Letting $t_{i}$ be the year
 paper $i$ was published and $n_{jt}$ be the total number of
  publications of researcher $j$ up to year $t$, we used the
  following model
\begin{align*}
y_{ij}\overset{ind}{\sim} & \mbox{Binomial}\left(n_{jt_i},
                            \mbox{logit}^{-1}(x_{ij})\right)\mbox{ for }j=1,\ldots,p,
  \\
\bm{x}_{i} \overset{ind}{\sim} & N(\bm{\mu}_{t_{i}},\bm{\Sigma}),\quad
                                 \bm{\mu}_{t} \overset{iid}{\sim} N(\bm{\mu}_{0},\bm{D}),
\end{align*}
where $\bm{x}_{i}=\big(x_{i1},\ldots,x_{ip}\big)^{\top}$,
$p=11$ and $\bm{D}$ is a diagonal matrix with the diagonal elements
being positive. The model parameters are $\bm{\mu}_{0}$,
$\bm{\Sigma}$ and $\bm{D}$ whereas $\bm{\mu}_{t}$'s are
random effects.

\begin{figure} 
\centering
     \includegraphics[width=0.6\textwidth]{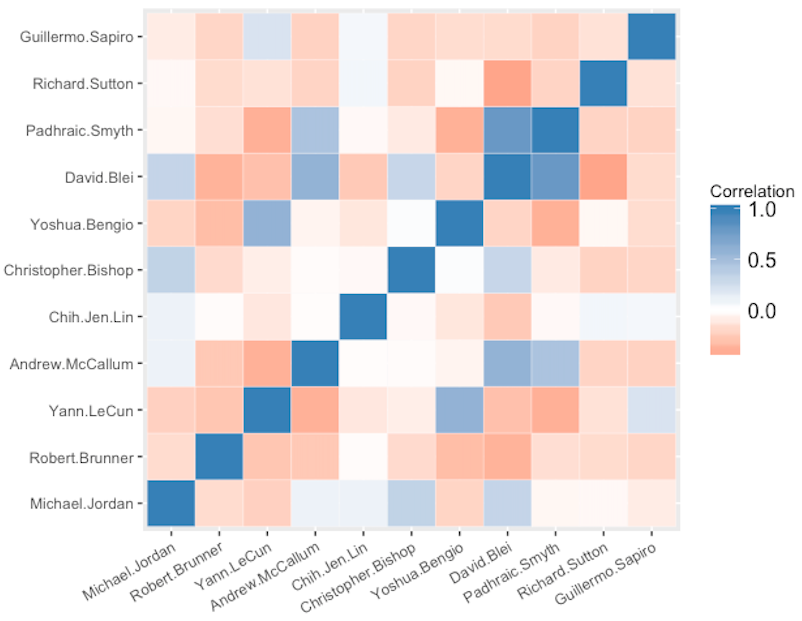}
\caption{Visualizing the correlation matrix induced by $\bm{\Sigma}$.}
\label{fig:corr}
\end{figure}

After integrating out $\bm{\mu}_{t}$'s and $\bm{x}_{i}$'s,
$\bm{y}_{i}$'s are no longer independent. It is easy to check that the
above model is {\em mosaic-type} in the generalized {\em Bayesian
  mosaic} framework of \S\ref{ssec:generalize}. Normal random walk MH sampler was implemented for sampling the {\em
knot marginals} for 40000 iterations with the
first 20000 as burn-in and thinning the rest into 500 
samples.  In sampling the {\em tiles}, we
used the plug-in approach discussed in \S\ref{ssec:sampling} and
sampled from the resulted {\em tile conditionals} via MH. For each
{\em tile},
we ran the MH
sampler 10000 iterations with the first 5000 as burn-in and thinned
the rest into 500 final samples. The entire sampling process took
around 5 hours with the jobs distributed to 11 parallel workers.

Figure~\ref{fig:mut} visualizes the posterior mean of the random
effects $\bm{\mu}_{t}$'s by different researchers. Intuitively, 
$\bm{\mu}_{t}$ is a vector of the average log-odds of a
single paper citing these researchers. Interestingly, while most of
the researchers' log-odds of being cited is decreasing, the only two exceptions are both working on deep learning.

We also computed the posterior mean of $\bm{\Sigma}$ (after
correction). The induced correlation matrix is visualized via a
heatmap in Figure~\ref{fig:corr}. Clearly, some researchers are
more likely to be cited together compared to the others,
indicating their strong overlapping research interests. For instance, Yann
Lecun and Yoshua Bengio have a stronger correlation since they
are both studying deep learning. There are researchers whose research
interests seem to overlap with many others, e.g., David Blei. Also,
there are researchers whose research interests seem to be unique in
this selected group,
e.g., Richard Sutton.

\appendix

\section{Proofs}
Whenever we write $\lim_{n\to\infty}a_{n}=a_{0}$, we mean the limit
holds with $P_{\bm{\theta}_{0}}$-probability one. We will omit the
phrase ``with $P_{\bm{\theta}_{0}}$-probability one'' for
succinctness.

\subsection{Proof of Lemma~\ref{lem:mosaic-rich}}
Since $P_{\bm{\psi}}$ is {\em
  mosaic-type}, from Definition~\ref{def:mosaic-type}, we have $\bm{\psi}_{st}$, $1\leq t\leq s \leq
p$ such that 
\[\bm{\psi}=\big[\bm{\psi}_{12}^{\top},\ldots,
  \bm{\psi}_{(p-1)p}^{\top}, \bm{\psi}_{11}^{\top},\ldots,
  \bm{\psi}_{pp}^{\top}\big]^{\top}.\]
For $j=1,\ldots,p$, the density of the univariate marginal
  distribution of $P_{\bm{\psi}}$ is
$f_{0,jj}(x_{j}|\bm{\psi}_{jj})$.
And for $1 \leq t < s \leq p$, the density of the bivariate
  marginal data distribution of $P_{\bm{\psi}}$ is
$f_{0,st}(x_{s},x_{t}|\bm{\psi}_{st}, \bm{\psi}_{ss}, \bm{\psi}_{tt})$.
Introduce
\[\bm{\theta}=\big[\bm{\psi}_{12}^{\top},\ldots,
  \bm{\psi}_{(p-1)p}^{\top}, \bm{\psi}_{11}^{\top}, \bm{\mu}_{1}^{\top},\ldots,
  \bm{\psi}_{pp}^{\top}, \bm{\mu}_{p}^{\top}\big]^{\top},\]
and let $\bm{\theta}_{st}=\bm{\psi}_{st}$ for $1\leq s < t \leq p$ and $\bm{\theta}_{jj}=\left(\bm{\psi}_{jj}^{\top}, \bm{\mu}_{j}^{\top}\right)^{\top}$ for $j=1,\ldots,p$.
From (\ref{lem:mosaic-rich}), it can be shown that
\begin{align*}
& \int\cdots\int f(\bm{y}|\bm{\mu},\bm{\psi})\intd y_{2}\cdots\intd
  y_{p} \\
= & \int
  f_{0}(\bm{x}|\bm{\psi}) g_{1}(y_{1}|x_{1},\bm{\mu}_{1})
    \prod_{j=2}^{p}\bigg[\int g_{j}(y_{j}|x_{j},\bm{\mu}_{j}) \intd
    y_{j}\bigg]\intd \bm{x} \\
= & \int
  f_{0,11}(x_{1}|\bm{\psi}_{11}) g_{1}(y_{1}|x_{1},\bm{\mu}_{1}) \intd
    x_{1} \\
= & f_{11}(y_{1}|\bm{\theta}_{11}).
\end{align*}
Similarly, one can show that there exists a collection of density
functions $\left\{f_{jj}(y_{j}|\bm{\theta}_{jj})\right\}$ such that
for $j=2,\ldots,p$, the density of the univariate marginal distribution of $P_{\bm{\mu},\bm{\psi}}$ is
$f_{jj}(y_{j}|\bm{\theta}_{jj})$. And that there exists a collection of density
functions $\left\{f_{st}(y_{s},y_{t}|\bm{\theta}_{st}, \bm{\theta}_{ss},
  \bm{\theta}_{tt})\right\}$ such that
for $1 \leq t < s \leq p$, the density of the bivariate
  marginal data distribution of $P_{\bm{\mu},\bm{\psi}}$ is
$f_{st}(y_{s},y_{t}|\bm{\theta}_{st}, \bm{\theta}_{ss},
  \bm{\theta}_{tt})$.
Hence $P_{\bm{\mu},\bm{\psi}}$ is also {\em Mosaic-type}.

\subsection{Taylor Expansions \& Upper Bounds}
We will find
the limit and derive an upper
bound for the Taylor expansion of  $L_{n}(\bm{\eta},\bm{\zeta})$, which will be
used in later proofs.

\begin{lemma} 
\label{lem:quad}
Letting $\bm{\alpha}$ be a
$d$-dimensional {\em multi-index} for $\bm{x}$, consider
$M_{\bm{\alpha}}<\infty$ for all $\bm{\alpha}$ such that
$|\bm{\alpha}|=3$. Then for any positive definite matrix
$\bm{\Lambda}$, we can find $\delta>0$ such that when 
$|\bm{x}|<\sqrt{n}\delta$,
\begin{align*}
\bigg|\sum_{|\bm{\alpha}|=3}M_{\bm{\alpha}}\frac{\bm{x}^{\bm{\alpha}}}{\sqrt{n}}\bigg| < \frac{1}{2}\bm{x}^{\top}\bm{\Lambda}\bm{x}.
\end{align*}
\end{lemma}

\begin{proof}
It is easily seen that
\begin{align*}
\bigg|\sum_{|\bm{\alpha}|=3}M_{\bm{\alpha}}\frac{\bm{x}^{\bm{\alpha}}}{\sqrt{n}}\bigg|
  < \sum_{|\bm{\alpha}|=3}M_{\bm{\alpha}}\bigg|\frac{\bm{x}^{\bm{\alpha}}}{\sqrt{n}}\bigg|.
\end{align*}
Consider $\bm{\alpha}=(\alpha_1,\ldots,\alpha_d)$ where $|\bm{\alpha}|=3$, and suppose $\alpha_{j_1}$
is the first non-zero index. Since
$|\bm{x}|<\sqrt{n}\delta$, we have
$\frac{|x_{j_1}|}{\sqrt{n}}<\delta$. 
Assume that $\alpha_{j_2}$ and $\alpha_{j_3}$ are the other two
non-zero indices, note that we allow $j_{2}=j_{3}$. We have
\begin{align*}
\bigg|\frac{\bm{x}^{\bm{\alpha}}}{\sqrt{n}}\bigg| <
  \delta|x_{j_2}x_{j_3}| \leq \frac{\delta}{2}\big(x^2_{j_2}+x^2_{j_3}\big).
\end{align*}
Doing this for all $\bm{\alpha}$ such that $|\bm{\alpha}|=3$, it can be shown that
\begin{align*}
\sum_{|\bm{\alpha}|=3}M_{\bm{\alpha}}\bigg|\frac{\bm{x}^{\bm{\alpha}}}{\sqrt{n}}\bigg|
  < \frac{\delta}{2}\bm{x}^{\top}\bm{\Psi}\bm{x},
\end{align*}
where $\bm{\Psi}$ is some diagonal matrix with all diagonal elements
being positive.
Since $\bm{\Lambda}$ is positive definite, we can always find 
$\delta>0$ so that $\bm{\Lambda}-\delta\bm{\Psi}$ is also
positive definite, which implies that for any $\bm{x}\in\setr^{d}$,
\begin{align*}
\frac{1}{2}\bm{x}^{\top}\bm{\Lambda}\bm{x} -
  \frac{\delta}{2}\bm{x}^{\top}\bm{\Psi}\bm{x} =
  \frac{1}{2}\bm{x}^{\top}\big(\bm{\Lambda}-\delta
  \bm{\Psi}\big)\bm{x} > 0.
\end{align*}
\end{proof}

Recall that $\bm{t}=\sqrt{n}\big(\bm{\eta}-\hat{\bm{\eta}}_{n}\big)$
and $\bm{r}=\sqrt{n}\big(\bm{\zeta}-\hat{\bm{\zeta}}_{n}\big)$. Letting
$\bm{\alpha}$ be a $d_{\eta}$-dimensional {\em multi-index} for 
$\bm{t}$ and $\bm{\beta}$ be a $d_{\zeta}$-dimensional {\em multi-index} for
$\bm{r}$, we expand the Taylor series for $L_{n}(\hat{\bm{\eta}}_{n}+\bm{t}/\sqrt{n},
                                              \hat{\bm{\zeta}}_{n}+\bm{r}/\sqrt{n})
                                              -
                                              L_{n}(\hat{\bm{\eta}}_{n},
                                              \hat{\bm{\zeta}}_{n})$ and get
\begin{align}
&L_{n}(\hat{\bm{\eta}}_{n}+\bm{t}/\sqrt{n},
                                              \hat{\bm{\zeta}}_{n}+\bm{r}/\sqrt{n})
                                              -
                                              L_{n}(\hat{\bm{\eta}}_{n},
                                              \hat{\bm{\zeta}}_{n})
  \label{eq:expansion-theta1}\\
= & \frac{1}{2}\begin{brsm}
    \bm{t}\\
    \bm{r}
\end{brsm}^{\top}\frac{1}{n}L_{n}^{(2)}(\hat{\bm{\eta}}_{n},
                                              \hat{\bm{\zeta}}_{n}) \begin{brsm}
    \bm{t}\\
    \bm{r}
\end{brsm} +R_{n}\left(\bm{t},\bm{r}\right), \nonumber
\end{align}
where $\bm{\eta}^{'}_{n}$ is between
$\hat{\bm{\eta}}_{n}+\bm{t}/\sqrt{n}$ and
$\hat{\bm{\eta}}_{n}$, $\bm{\zeta}^{'}_{n}$ is between
$\hat{\bm{\zeta}}_{n}+\bm{r}/\sqrt{n}$ and $\hat{\bm{\zeta}}_{n}$ and
\begin{align*}
R_{n}\left(\bm{t},\bm{r}\right)=\sum_{|\bm{\alpha}|+|\bm{\beta}|=3}\frac{1}{n\bm{\alpha}!}\partial^{\bm{\alpha}}\partial^{\bm{\beta}}L_{n}(\bm{\eta}^{'}_{n},\bm{\zeta}^{'}_{n})\frac{\bm{t}^{\bm{\alpha}}\bm{r}^{\bm{\beta}}}{\sqrt{n}}.
\end{align*}
Letting $\hat{\bm{I}}_n = -\frac{1}{n}L_{n}^{(2)}(\hat{\bm{\eta}}_{n},
                                              \hat{\bm{\lambda}}_n)$
and $\hat{\bm{I}}_n = 
\begin{brsm}
   \hat{\bm{I}}^{11}_n & \hat{\bm{I}}^{12}_n\\
   \hat{\bm{I}}^{21}_n & \hat{\bm{I}}^{22}_n
\end{brsm}$, (\ref{eq:expansion-theta1}) can be written as
\begin{align}
\label{eq:expansion-theta2}
\begin{split}
&L_{n}(\hat{\bm{\eta}}_{n}+\bm{t}/\sqrt{n},
                                              \hat{\bm{\zeta}}_{n}+\bm{r}/\sqrt{n})
                                              -
                                              L_{n}(\hat{\bm{\eta}}_{n},
                                              \hat{\bm{\zeta}}_{n})\\
= & 
-\frac{1}{2}\bm{t}^{\top}\hat{\bm{I}}^{11}_n\bm{t}-\frac{1}{2}\bm{r}^{\top}\hat{\bm{I}}^{22}_n\bm{r}-\bm{t}^{\top}\hat{\bm{I}}^{12}_n\bm{r}+
R_{n}(\bm{t},\bm{r}).
\end{split}
\end{align}

\begin{corollary}
\label{coro:remainder2}
If conditions~\ref{cond:1-1}-\ref{cond:1-5} hold for $f_{2}(\bm{y}|\bm{\eta},\bm{\zeta})$,
then $\exists \delta>0$ such that for all $\bm{t}$ and $\bm{r}$ satisfying
$\left\|\begin{brsm}
    \bm{t}/\sqrt{n}\\
    \bm{r}/\sqrt{n}
\end{brsm}\right\|<\delta$,
the followings are true:
\begin{itemize}
\item[i)] For any fixed $\bm{t}$ and $\bm{r}$,
$\lim_{n\to\infty}R_{n}(\bm{t},\bm{r})=0$.
\item[ii)] For any positive definite matrix $\bm{\Lambda}$,
$\exists N\in\setnnn$ such that
$\forall n>N$,
\begin{align}
\label{eq: remainder1-1}
\big|R_{n}(\bm{t},\bm{r})\big|
  < \begin{brsm}
    \bm{t}\\
    \bm{r}
\end{brsm}^{\top}\bm{\Lambda}\begin{brsm}
    \bm{t}\\
    \bm{r}
\end{brsm}.
\end{align}
\end{itemize}
\end{corollary}

\begin{proof}
From condition~\ref{cond:1-2},  for some $\delta^{'}>0$ we have
\[\sup_{\bm{z}\in \mathcal{B}(\bm{z}_0, \delta^{'})} \left|\partial^{\bm{\alpha}}\ell_{2}(\bm{\eta},\bm{\zeta},
\bm{y})\right|\leq
  M_{\bm{\alpha},2}(\bm{y}),\]
and $\mathbb{E}_{\bm{\theta}_{0}}M_{\bm{\alpha},2}(\bm{y})<\infty$. Letting $\delta<\delta^{'}$, we
have 
\begin{align*}
\left|\sum_{|\bm{\alpha}|+|\bm{\beta}|=3}\frac{1}{n\bm{\alpha}!}\partial^{\bm{\alpha}}\partial^{\bm{\beta}}L_{n}(\bm{\eta}^{'}_{n},\bm{\zeta}^{'}_{n})\frac{\bm{t}^{\bm{\alpha}}\bm{r}^{\bm{\beta}}}{\sqrt{n}}\right|< \sum_{|\bm{\alpha}|+|\bm{\beta}|=3}\frac{1}{n\bm{\alpha}!}\sum_{i=1}^{n}M_{\bm{\alpha},2}(\bm{y}_{i}) \frac{\bm{t}^{\bm{\alpha}}\bm{r}^{\bm{\beta}}}{\sqrt{n}}.
\end{align*}
From strong law of large numbers (SLLN) we know that $\lim_{n\to\infty}\frac{1}{n}\sum_{i=1}^{n}M_{\bm{\alpha},2}(\bm{y}_{i}) =\mathbb{E}_{\bm{\theta}_{0}}M_{\bm{\alpha},2}(\bm{y})$,
hence $\exists N_1\in\setnnn$ such that $\forall n>N_{1}$,
\begin{align*}
\left|\sum_{|\bm{\alpha}|+|\bm{\beta}|=3}\frac{1}{n\bm{\alpha}!}\partial^{\bm{\alpha}}\partial^{\bm{\beta}}L_{n}(\bm{\eta}^{'}_{n},\bm{\zeta}^{'}_{n})\frac{\bm{t}^{\bm{\alpha}}\bm{r}^{\bm{\beta}}}{\sqrt{n}}\right|< 2\sum_{|\bm{\alpha}|+|\bm{\beta}|=3}\frac{1}{\bm{\alpha}!}\mathbb{E}_{\bm{\theta}_{0}}M_{\bm{\alpha},2}(\bm{y}) \frac{\bm{t}^{\bm{\alpha}}\bm{r}^{\bm{\beta}}}{\sqrt{n}}.
\end{align*}
Apparently, for fixed $\bm{t}$ and $\bm{r}$,
\[\lim_{n\to\infty}\left|R_{n}\left(\bm{t},\bm{r}\right)\right|<\lim_{n\to\infty}2\sum_{|\bm{\alpha}|+|\bm{\beta}|=3}\frac{1}{\bm{\alpha}!}\mathbb{E}_{\bm{\theta}_{0}}M_{\bm{\alpha},2}(\bm{y}) \frac{\bm{t}^{\bm{\alpha}}\bm{r}^{\bm{\beta}}}{\sqrt{n}}=0.\]
Therefore $\lim_{n\to\infty}\left|R_{n}\left(\bm{t},\bm{r}\right)\right|=0$.

Applying Lemma~\ref{lem:quad}, we could find
$\delta<\delta^{'}$ and $N_{2}\in\setnnn$ such that $\forall
n>\max\left\{N_{1},N_{2}\right\}$,
\begin{align*}
\left|2\sum_{|\bm{\alpha}|+|\bm{\beta}|=3}\frac{1}{\bm{\alpha}!}\mathbb{E}_{\bm{\theta}_{0}}M_{\bm{\alpha},2}(\bm{y}) \frac{\bm{t}^{\bm{\alpha}}\bm{r}^{\bm{\beta}}}{\sqrt{n}}\right|< \begin{brsm}
    \bm{t}\\
    \bm{r}
\end{brsm}^{\top}\bm{\Lambda}\begin{brsm}
    \bm{t}\\
    \bm{r}
\end{brsm}.
\end{align*}
Letting $N=\max\left\{N_{1},N_{2}\right\}$ finishes the proof.
\end{proof}

\subsection{Proof of Lemma~\ref{lem:bi-dist-dct}}
For $\delta>0$ , we define
$A_{n,1}(\delta)=\{\bm{r}:\|\bm{r}\|<\sqrt{n}\delta\}$,
$A_{n,2}(\delta)=\{\bm{r}:\|\bm{r}\|>\sqrt{n}\delta\}$,
$B_{n,1}(\delta)=\{\bm{t}:\|\bm{t}\|<\sqrt{n}\delta\}$ and
$B_{n,2}(\delta)=\{\bm{t}:\|\bm{t}\|>\sqrt{n}\delta\}$. We first provide
a lemma that will be used in our later proof of the main
result.
\begin{lemma} \label{lem:bi-dist-dct}
Suppose that conditions~\ref{cond:1-1}-\ref{cond:1-5} hold for $f_{2}(\bm{y}|\bm{\eta},\bm{\zeta})$ and the
prior density $\pi(\bm{\eta}|\bm{\zeta})$ is continuous and positive
at $\begin{brsm}
\bm{\eta}_{0} \\
    \bm{\zeta}_{0}
\end{brsm}$, then $\exists \delta_{t}>0$, $\delta_{r}>0$ and $N\in\setnnn$ such
that the followings
are true:
\begin{itemize}
\item For any fixed $\bm{t}$ and $\bm{r}$, with $P_{\bm{\theta}_{0}}$-probability one
\begin{align} \label{eq: lem-bi-dist-dct-limit}
\begin{split}
& \lim_{n=\infty}\pi_{n,2}^{*}(\bm{t}|\bm{r})\mathbbm{1}\{\bm{r}\in
  A_{n,1}(\delta_{r}),\bm{t}\in B_{n,1}(\delta_{t})\}\\
=&\phi\left(\bm{t}\middle|-\big(\bm{I}^{11}_{0}\big)^{-1}\bm{I}^{12}_{0}\bm{r},\big(\bm{I}^{11}_{0}\big)^{-1}\right).
\end{split}
\end{align}

\item $\exists \epsilon>0$ and $c(\delta_{r},\delta_{t})>0$ such that
\begin{align} \label{eq: lem-bi-dist-dct-tozero}
\pi^{*}_{n,2}(\bm{t}|\bm{r}) \mathbbm{1}\{\bm{r}\in
  A_{n,1}(\delta_{r}),\bm{t}\in B_{n,2}(\delta_{t})\}< & \frac{\left|4\pi\bm{I}_{0}^{11}\right|^{-1/2}}{c(\delta_{r},\delta_{t})}\exp\{-n\epsilon\}.
\end{align}
\end{itemize}
\end{lemma}

\begin{proof}
The proof consists of the following four steps.

\noindent{\bf Step 1}~~~~~~~~ In this step we will find the limit of the
normalizing constant. The constant is
\begin{align}
\label{eq:tn-anr-step3}
a_{n}(\bm{r}) =&\int g_{n}(\bm{t},\bm{r})\intd\bm{t},
\end{align}
where
\begin{align*}
g_{n}(\bm{t},\bm{r}) =&\exp\big\{L_{n}\big(\hat{\bm{\eta}}_{n}+\bm{t}/\sqrt{n}, \hat{\bm{\zeta}}_{n}+\bm{r}/\sqrt{n}\big) - L_{n}\big(\hat{\bm{\eta}}_{n}, \hat{\bm{\zeta}}_{n}\big)\big\}
\\
& \qquad \times\pi\big(\hat{\bm{\eta}}_{n}+\bm{t}/\sqrt{n}\big|\hat{\bm{\zeta}}_{n}+\bm{r}/\sqrt{n}\big) \mathbbm{1}\{\bm{r}\in A_{n,1}(\delta_{r})\}.
\end{align*}
We can find the limit of $a_n(\bm{r})$ by finding the limits of
$\int g_{n}(\bm{t},\bm{r}) \mathbbm{1}\{\bm{t}\in B_{n,1}(\delta_{t})\}\intd\bm{t}$
and of
$\int g_{n}(\bm{t},\bm{r}) \mathbbm{1}\{\bm{t}\in
B_{n,2}(\delta_{t})\}\intd\bm{t}$, since $a_n(\bm{r})$ is the sum of the these two integrals. We start with the first one.

The Taylor expansion of $L_{n}\big(\hat{\bm{\eta}}_{n}+\bm{t}/\sqrt{n}, \hat{\bm{\zeta}}_{n}+\bm{r}/\sqrt{n}\big) - L_{n}\big(\hat{\bm{\eta}}_{n}, \hat{\bm{\zeta}}_{n}\big)$ is given in (\ref{eq:expansion-theta2}). Applying Corollary~\ref{coro:remainder2}, and since
$\big\|\begin{brsm}
    \bm{t}\\
    \bm{r}
\end{brsm}\big\| 
\leq \|\bm{r}\| + \|\bm{t}\|$,
we could find
$\delta_t>0$ and $N_{11}\in\setnnn$ such that $\forall \bm{t}\in B_{n,1}(\delta_{t})$, $\forall \bm{r}\in A_{n,1}(\delta_{r})$ and $\forall n>
N_{11}$,
\begin{align}
\label{eq:tn-ub1}
\big|R_{n}(\bm{t},\bm{r})\big|
  < \frac{1}{8}\begin{brsm}
    \bm{t}\\
    \bm{r}
\end{brsm}^{\top}\bm{I}_{0}\begin{brsm}
    \bm{t}\\
    \bm{r}
\end{brsm}.
\end{align}
From condition~\ref{cond:1-5}, we know that
$\lim_{n\to\infty}\hat{\bm{\eta}}_{n}=\bm{\eta}_{0}$ and
$\lim_{n\to\infty}\hat{\bm{\zeta}}_{n}=\bm{\zeta}_{0}$. Applying
condition~\ref{cond:1-3}, we can show that
$\lim_{n\to\infty}\hat{\bm{I}}_{n}=\bm{I}_{0}$. Moreover, for any fixed $\bm{t}$ and
$\bm{r}$,
we know that $\lim_{n\to\infty} R_{n}(\bm{t},\bm{r})=0$
(Corollary~\ref{coro:remainder2}). Therefore,  
\begin{align}
\label{eq:tn-limit-step3}
\begin{split}
& \lim_{n\to\infty}\exp\big\{L_{n}(\hat{\bm{\eta}}_{n}+\bm{t}/\sqrt{n},
  \hat{\bm{\zeta}}_{n}+\bm{r}/\sqrt{n}) -
  L_{n}(\hat{\bm{\eta}}_{n}, \hat{\bm{\zeta}}_{n})\big\} \\
  = &\exp\bigg\{-\frac{1}{2}\begin{brsm}
    \bm{t}\\
    \bm{r}
\end{brsm}^{\top}\bm{I}_{0}\begin{brsm}
    \bm{t}\\
    \bm{r}
\end{brsm}\bigg\} = \exp\left\{-\frac{1}{2}\bm{t}^{\top}\bm{I}^{11}_{0}\bm{t}-\frac{1}{2}\bm{r}^{\top}\bm{I}^{22}_{0}\bm{r}-\bm{t}^{\top}\bm{I}^{12}_{0}\bm{r}\right\},
\end{split}
\end{align}
where $\bm{I}_{0} = 
\begin{brsm}
   \bm{I}^{11}_{0} & \bm{I}^{12}_{0}\\
   \bm{I}^{21}_{0} & \bm{I}^{22}_{0}
\end{brsm}$. 
Moreover, since $\pi(\bm{\eta}|\bm{\zeta})$ is positive and continuous at $\bm{\eta}=\bm{\eta}_{0}$ and $\bm{\zeta}=\bm{\zeta}_{0}$,
\begin{align}
\label{eq:limit-theta}
\begin{split}
&\lim_{n\to\infty} g_{n}(\bm{t},\bm{r}) \mathbbm{1}\{\bm{t}\in B_{n,1}(\delta_{t})\} 
  \\
=&  \exp\bigg\{-\frac{1}{2}\bm{t}^{\top}\bm{I}^{11}_{0}\bm{t}-\frac{1}{2}\bm{r}^{\top}\bm{I}^{22}_{0}\bm{r}-\bm{t}^{\top}\bm{I}^{12}_{0}\bm{r}\bigg\}\pi(\bm{\eta}_{0}|\bm{\zeta}_{0}).
\end{split}
\end{align}

From (\ref{eq:tn-limit-step3}), we could find $N_{12}\in\setnnn$ such that $\forall
n>N_{12}$,
\begin{align}
\label{eq:tn-ub2}
& \exp\big\{L_{n}(\hat{\bm{\eta}}_{n}+\bm{t}/\sqrt{n},
  \hat{\bm{\zeta}}_{n}+\bm{r}/\sqrt{n}) -
  L_{n}(\hat{\bm{\eta}}_{n}, \hat{\bm{\zeta}}_{n})\big\}
  < \exp\bigg\{-\frac{1}{4}\begin{brsm}
    \bm{t}\\
    \bm{r}
\end{brsm}^{\top}\bm{I}_{0}\begin{brsm}
    \bm{t}\\
    \bm{r}
\end{brsm}\bigg\}
\end{align}
Let $N_1=\max\{N_{11},N_{12}\}$. Combining (\ref{eq:tn-ub1}) and
(\ref{eq:tn-ub2}) we have, $\forall \bm{t}\in B_{n,1}(\delta_{t})$, $\forall \bm{r}\in A_{n,1}(\delta_{r})$ and $\forall n>N_1$,
\begin{align*}
& \exp\big\{L_{n}(\hat{\bm{\eta}}_{n}+\bm{t}/\sqrt{n},
  \hat{\bm{\zeta}}_{n}+\bm{r}/\sqrt{n}) -
  L_{n}(\hat{\bm{\eta}}_{n}, \hat{\bm{\zeta}}_{n})\big\}
  < \exp\bigg\{-\frac{1}{8}\begin{brsm}
    \bm{t}\\
    \bm{r}
\end{brsm}^{\top}\bm{I}_{0}\begin{brsm}
    \bm{t}\\
    \bm{r}
\end{brsm}\bigg\}
\end{align*}
Let $b(\delta_{t},\delta_{r})=\sup_{\|\bm{\eta}-\bm{\eta}_{0}\|<2\delta_{t}, \|\bm{\zeta}-\bm{\zeta}_{0}\|<2\delta_{r}}\pi(\bm{\eta}|\bm{\zeta})$. Given that $\pi(\bm{\eta}|\bm{\zeta})$ is positive and continuous at
$\bm{\eta}=\bm{\eta}_{0}$ and $\bm{\zeta}=\bm{\zeta}_{0}$, we can
choose $\delta_{r}$ and $\delta_{t}$ small enough so that $b(\delta_{t},\delta_{r})>0$. Then $g_{n}(\bm{t},\bm{r}) \mathbbm{1}\{\bm{t}\in B_{n,1}(\delta_{t})\} $ is bounded by
\begin{align*}
& b(\delta_{t},\delta_{r})\int_{B_{n,1}(\delta_{t})}\exp\bigg\{-\frac{1}{8}\begin{brsm}
    \bm{t}\\
    \bm{r}
\end{brsm}^{\top}\bm{I}_{0}\begin{brsm}
    \bm{t}\\
    \bm{r}
\end{brsm}\bigg\} \intd\bm{t},
\end{align*}
which is clearly integrable. Applying DCT,
\begin{align*}
&\lim_{n\to\infty}\int g_{n}(\bm{t},\bm{r}) \mathbbm{1}\{\bm{t}\in B_{n,1}(\delta_{t})\} 
  \intd\bm{t} \\
= &
    \int\exp\bigg\{-\frac{1}{2}\bm{t}^{\top}\bm{I}^{11}_{0}\bm{t}-\frac{1}{2}\bm{r}^{\top}\bm{I}^{22}_{0}\bm{r}-\bm{t}^{\top}\bm{I}^{12}_{0}\bm{r}\bigg\}\pi(\bm{\eta}_{0}|\bm{\zeta}_{0})\intd\bm{t}
  \\
= & \exp\bigg\{-\frac{1}{2}\bm{r}^{\top}\bigg(\bm{I}^{22}_{0}-\bm{I}^{21}_{0}\big(\bm{I}^{11}_{0}\big)^{-1}\bm{I}^{12}_{0}\bigg)\bm{r}\bigg\}\pi(\bm{\eta}_{0}|\bm{\zeta}_{0})|\bm{I}^{11}_{0}/2\pi|^{-1/2}
\end{align*}

We complete the this step by finding the limit for $\int g_{n}(\bm{t},\bm{r})
\mathbbm{1}\{\bm{t}\in B_{n,2}(\delta_{t})\}\intd\bm{t}$. Similar to Step
3 in the
proof of Theorem~\ref{thm:joint-normality}, we can find $\epsilon>0$ and
$N_{2}\in\setnnn$ such that $\forall n>N_2$, 
\begin{align*}
\frac{1}{n}\bigg[L_{n}(\hat{\bm{\eta}}_{n}+\bm{t}/\sqrt{n},
  \hat{\bm{\zeta}}_{n}+\bm{r}/\sqrt{n}) -
  L_{n}(\hat{\bm{\eta}}_{n}, \hat{\bm{\zeta}}_{n})\bigg] < -\epsilon.
\end{align*}
This implies
\begin{align*}
&\lim_{n\to\infty}\int g_{n}(\bm{t},\bm{r}) \mathbbm{1}\{\bm{t}\in B_{n,2}(\delta_{t})\}\intd\bm{t} \\
\leq & \lim_{n\to\infty}\exp\big\{-n\epsilon\big\}\int\pi(\hat{\bm{\eta}}_{n}+\bm{t}/\sqrt{n}|\hat{\bm{\zeta}}_{n}+\bm{r}/\sqrt{n})
  \mathbbm{1}\{\bm{t}\in B_{n,2}(\delta_{t})\}\intd\bm{t} \\
\leq & \lim_{n\to\infty}\exp\big\{-n\epsilon\big\}\\
= & 0.
\end{align*}
Hence we have shown that $\forall\bm{r}\in A_{n,2}(\delta_{r})$,
\begin{align}
\label{eq:limit-anr}
\lim_{n\to\infty}a_n(\bm{r})= &\exp\bigg\{-\frac{1}{2}\bm{r}^{\top}\bigg(\bm{I}^{22}_{0}-\bm{I}^{21}_{0}\big(\bm{I}^{11}_{0}\big)^{-1}\bm{I}^{12}_{0}\bigg)\bm{r}\bigg\}\pi(\bm{\eta}_{0}|\bm{\zeta}_{0})|\bm{I}^{11}_{0}/2\pi|^{-1/2}
\end{align}

\noindent{\bf Step 2}~~~~~~~~ In this step we will find the limit for
$\pi_{n}^{*}(\bm{t}|\bm{r})\mathbbm{1}\{\bm{r}\in A_{n,1}(\delta_{r})\}\mathbbm{1}\{\bm{t}\in B_{n,1}(\delta_{t})\}$. Since
\[\pi_{n}^{*}(\bm{t}|\bm{r})\mathbbm{1}\{\bm{r}\in A_{n,1}(\delta_{r})\}\mathbbm{1}\{\bm{t}\in B_{n,1}(\delta_{t})\} = a_{n}(\bm{r})^{-1}g_{n}(\bm{t},\bm{r})
  \mathbbm{1}\{\bm{t}\in B_{n,1}(\delta_{t})\},\]
combining (\ref{eq:limit-theta}) and (\ref{eq:limit-anr}) we
immediately have
\begin{align*}
 &\lim_{n\to\infty}\pi_{n}^{*}(\bm{t}|\bm{r})\mathbbm{1}\{\bm{r}\in
   A_{n,1}(\delta_{r})\}\mathbbm{1}\{\bm{t}\in B_{n,1}(\delta_{r})\}\\
= &
                                                \lim_{n\to\infty}a_n(\bm{r}) \lim_{n\to\infty}g_{n}(\bm{t},\bm{r})
  \mathbbm{1}\{\bm{t}\in B_{n,1}(\delta_{t})\} \\
= & |\bm{I}^{11}_{0}/2\pi|^{1/2}\exp\{-\frac{1}{2}\bigg[\bm{t}-\big(\bm{I}^{11}_{0}\big)^{-1}\bm{I}^{12}_{0}\bm{r}\bigg]^{\top}\bm{I}^{11}_{0}\bigg[\bm{t}-\big(\bm{I}^{11}_{0}\big)^{-1}\bm{I}^{12}_{0}\bm{r}\bigg]\}\\
= & \phi\bigg(\bm{t}\bigg|-\big(\bm{I}^{11}_{0}\big)^{-1}\bm{I}^{12}_{0}\bm{r},\big(\bm{I}^{11}_{0}\big)^{-1}\bigg).
\end{align*}

\noindent{\bf Step 3}~~~~~~~~ In this step, we complete the proof by finding a lower bound for
$a_{n}(\bm{r})$.
Applying
Corollary~\ref{coro:remainder2} one more time, by choosing
$\delta_{t}$ and $\delta_{r}$ small enough, $\exists
N_{31}\in\setnnn$ such that $\forall \bm{t}\in B_{n,1}(\delta_{t})$, $\forall \bm{r}\in A_{n,1}(\delta_{r})$ and $\forall n>
N_{21}$,
\begin{align}
\label{eq:tn-ub3}
\big|R_{n}(\bm{t},\bm{r})\big|
  < \frac{1}{4}\begin{brsm}
    \bm{t}\\
    \bm{r}
\end{brsm}^{\top}
\bm{I}_{0}
\begin{brsm}
    \bm{t}\\
    \bm{r}
\end{brsm}.
\end{align}
Also from (\ref{eq:tn-limit-step3}), $\exists N_{32}\in\setnnn$ such that $\forall
n>N_{22}$,
\begin{align*}
\begin{split}
& \exp\big\{L_{n}(\hat{\bm{\eta}}_{n}+\bm{t}/\sqrt{n},
  \hat{\bm{\zeta}}_{n}+\bm{r}/\sqrt{n}) -
  L_{n}(\hat{\bm{\eta}}_{n}, \hat{\bm{\zeta}}_{n})\big\} \\
  > & \exp\bigg\{-\frac{1}{2}\begin{brsm}
    \bm{t}\\
    \bm{r}
\end{brsm}^{\top}
\big(\bm{I}_{0}+\frac{1}{2}\bm{I}_{0}\big)
\begin{brsm}
    \bm{t}\\
    \bm{r}
\end{brsm}\bigg\}.
\end{split}
\end{align*}
Combining above and (\ref{eq:tn-ub3}), we immediately have that $\forall \bm{r}\in A_{n,1}(\delta_{r})$,
$\forall \bm{t}\in B_{n,1}(\delta_{t})$ and $\forall n>\max\{N_{31},N_{32}\}$,
\begin{align}
\label{eq:tn-step3-lb}
\begin{split}
& \exp\big\{L_{n}(\hat{\bm{\eta}}_{n}+\bm{t}/\sqrt{n},
  \hat{\bm{\zeta}}_{n}+\bm{r}/\sqrt{n}) -
  L_{n}(\hat{\bm{\eta}}_{n}, \hat{\bm{\zeta}}_{n})\big\} 
  > \exp\bigg\{-\begin{brsm}
    \bm{t}\\
    \bm{r}
\end{brsm}^{\top}\bm{I}_{0}
\begin{brsm}
    \bm{t}\\
    \bm{r}
\end{brsm}\bigg\}.
\end{split}
\end{align}
Let $c(\delta_{r},\delta_{t})=\inf_{\|\bm{\eta}-\bm{\eta}_{0}\|<2\delta_{t},
  \|\bm{\zeta}-\bm{\zeta}_{0}\|<2\delta_{r}}\pi(\bm{\eta}|\bm{\zeta})$. Given that $\pi(\bm{\eta}|\bm{\zeta})$ is positive and continuous at
$\bm{\eta}=\bm{\eta}_{0}$ and $\bm{\zeta}=\bm{\zeta}_{0}$, we can
choose $\delta_{r}$ and $\delta_{t}$ small enough so that $c(\delta_{r},\delta_{t})>0$.
Since $\lim_{n\to\infty}\hat{\bm{\eta}}_{n}= \bm{\eta}_{0}$ and
$\lim_{n\to\infty}\hat{\bm{\zeta}}_{n}= \bm{\zeta}_{0}$, there $\exists
N_{33}\in\setnnn$ such that $\forall \bm{r}\in A_{n,1}(\delta_{r})$ and $\forall \bm{t}\in
B_{n,1}(\delta_{t})$,
$\pi(\hat{\bm{\eta}}_{n}+\bm{t}/\sqrt{n}|\hat{\bm{\zeta}}_{n}+\bm{r}/\sqrt{n})>c(\delta_{r},\delta_{t})$. Letting $N_{3}=\max\{N_{31},N_{32},N_{33}\}$, together with
(\ref{eq:tn-step3-lb}), we have shown that $\forall n > N_{3}$, $a_n(\bm{r})$ is
lower-bounded by
\begin{align}
& c(\delta_{r},\delta_{t})\int\exp\bigg\{-\begin{brsm}
    \bm{t}\\
    \bm{r}
\end{brsm}^{\top}\bm{I}_{0}
\begin{brsm}
    \bm{t}\\
    \bm{r}
\end{brsm}\bigg\} \intd\bm{t} \nonumber\\
= &
    c(\delta_{r},\delta_{t})|4\pi\bm{I}_{0}^{11}|^{\frac{1}{2}}\exp\bigg\{-\bm{r}^{\top}\bigg(\bm{I}^{22}_{0}-\bm{I}^{21}_{0}\big(\bm{I}^{11}_{0}\big)^{-1}\bm{I}^{12}_{0}\bigg)\bm{r}\bigg\}. \label{eq:lem-bi-dist-dct-lb}
\end{align}

\noindent{\bf Step 4}~~~~~~~~Consider $\delta_{t}>0$. Since
$\lim_{n\to\infty}\hat{\bm{\eta}}_{n}= \bm{\eta}_{0}$
(condition~\ref{cond:1-5}), for $\forall\bm{t}\in B_{n,2}(\delta_{t})$ and $\forall\bm{r}\in A_{n,1}(\delta_{r})$, $\exists N_{1}\in\setnnn$ such that $\forall n>N_{1}$,
\begin{align*}
\big\|\begin{brsm}
    \hat{\bm{\eta}}_{n}+\bm{t}/\sqrt{n}-\bm{\eta}_{0}\\
    \hat{\bm{\zeta}}_{n}+\bm{t}/\sqrt{n}-\bm{\zeta}_{0}
\end{brsm}\big\|>\|\hat{\bm{\eta}}_{n}+\bm{t}/\sqrt{n}-\bm{\eta}_{0}\|>\frac{\delta_{t}}{2}.
\end{align*}
Applying
condition~\ref{cond:1-4}, $\exists
\epsilon>0$ and $N_{2}\geq N_{1}$ such that $\forall
\bm{t}\in B_{n,2}(\delta_{t})$, $\forall \bm{r}\in
A_{n,1}(\delta_{r})$ and $\forall n>N_{2}$,
\begin{align} \label{eq:tn-step4-b1}
\frac{1}{n}\big[L_{n}(\hat{\bm{\eta}}_{n}+\bm{t}/\sqrt{n},
  \hat{\bm{\zeta}}_{n}+\bm{r}/\sqrt{n}) -
  L_{n}(\bm{\eta}_{0}, \bm{\zeta}_{0})\big]< -3\epsilon.
\end{align}
From condition~\ref{cond:1-2}, $\exists N_{3}\in\setnnn$ such that
$\forall n>N_{3}$, $\begin{brsm}
    \hat{\bm{\eta}}_{n}\\
    \hat{\bm{\zeta}}_{n} 
\end{brsm}\in \mathcal{B}_{\bm{z}}\left(
\bm{z}_{0}
,\delta\right)$ where $\ell_{2}(\bm{\eta}, \bm{\zeta},
\bm{y})$ is thrice differentiable with respect to $\bm{\eta}$ and
$\bm{\zeta}$. Applying mean value theorem, we have
$\frac{1}{n}\big[L_{n}(\hat{\bm{\eta}}_{n},
  \hat{\bm{\zeta}}_{n}) - L_{n}(\bm{\eta}_{0}, \bm{\zeta}_{0})\big]=\frac{1}{n}L^{(1)}_{n}(\bm{\eta}^{'}_{n},
  \bm{\zeta}^{'}_{n})^{\top}
\begin{brsm}
    \hat{\bm{\eta}}_{n}-\bm{\eta}_{0}\\
    \hat{\bm{\zeta}}_{n} - \bm{\zeta}_{0}
\end{brsm}$, where $\bm{\eta}^{'}_{n}$ is between
$\hat{\bm{\eta}}_{n}$ and $\bm{\eta}_{0}$ and $\bm{\zeta}^{'}_{n}$ is
between $\hat{\bm{\zeta}}_{n}$ and $\bm{\zeta}_{0}$. Using
condition~\ref{cond:1-3} and continuous mapping theorem, it can be
shown that $\lim_{n\to\infty}\frac{1}{n}L^{(1)}_{n}(\bm{\eta}^{'}_{n},
  \bm{\zeta}^{'}_{n})=\bm{0}$. Hence, $\exists N_{4}>N_{3}$ such that
  $\forall n>N_{4}$,
$\frac{1}{n}\left| L_{n}(\hat{\bm{\eta}}_{n},
  \hat{\bm{\zeta}}_{n}) - L_{n}(\bm{\eta}_{0},
  \bm{\zeta}_{0})\right|<\epsilon$. Letting
$N=\max\{N_{1},\ldots,N_{4}\}$ and using
(\ref{eq:tn-step4-b1}), it can be shown that $\forall n>N$,
\begin{align} \label{eq:tn-step4-b3}
\frac{1}{n}\big[L_{n}(\hat{\bm{\eta}}_{n}+\bm{t}/\sqrt{n},
  \hat{\bm{\zeta}}_{n}+\bm{r}/\sqrt{n}) -
  L_{n}(\hat{\bm{\eta}}_{n},
  \hat{\bm{\zeta}}_{n})\big]< -2\epsilon.
\end{align}

Using Lemma~\ref{lem:bi-dist-dct}, we can choose $\delta_{t}$ and
$\delta_{r}$ to be small enough so that $a_{n}(\bm{r})\mathbbm{1}\{\bm{r}\in
A_{n,1}(\delta_{r})\}$ is lower bounded by
(\ref{eq:lem-bi-dist-dct-lb}). Since
$\bm{I}^{22}_{0}-\bm{I}^{21}_{0}\big(\bm{I}^{11}_{0}\big)^{-1}\bm{I}^{12}_{0}$
is positive definite, we can always choose $\delta_{r}$ small enough
so that 
\begin{align}
\label{eq:tn-step4-b2}
\exp\left\{-\bm{r}^{\top}\left(\bm{I}^{22}_{0}-\bm{I}^{21}_{0}\big(\bm{I}^{11}_{0}\big)^{-1}\bm{I}^{12}_{0}\right)\bm{r}\right\}
  < \exp\{n\epsilon\}.
\end{align}
Combining (\ref{eq:lem-bi-dist-dct-lb}), (\ref{eq:tn-step4-b3}) and
(\ref{eq:tn-step4-b2}), it can be shown that
\begin{align*}
pi^{*}_{n,2}(\bm{t}|\bm{r}) \mathbbm{1}\{\bm{r}\in
  A_{n,1}(\delta_{r}),\bm{t}\in B_{n,1}(\delta_{t})\}< & \frac{|4\pi\bm{I}_{0}^{11}|^{-1/2}}{c(\delta_{r},\delta_{t})}\exp\{-n\epsilon\}.
\end{align*}
\end{proof}

\subsection{Proof of Theorem~\ref{thm:joint-normality}}
We prove the theorem in the following four steps.

\noindent{\bf Step 1}~~~~~~~~
Consider
$\bm{r}=\sqrt{n}\big(\bm{\zeta}-\hat{\bm{\zeta}}_{n}\big)$ and its
posterior density $\pi_{n,1}^{*}(\bm{r})$. Applying
Lemma~\ref{lem:knot-normality} and a simple change of variable, we can show
\begin{align}
\lim_{n\to\infty}\int|\pi_{n,1}^{*}(\bm{r})-\phi\big(\bm{r}\big|\bm{\mu}_{n},
  \tilde{\bm{I}}_{0}^{-1}\big)\big|\intd\bm{r}= 0. \label{eq:knot-normality-coro}
\end{align}
It can show that the integral in (\ref{eq:joint-normality}) is bounded by
\begin{align*}
& \int\int\pi^{*}_{n,2}(\bm{t}|\bm{r})\big|\pi^{*}_{n,1}(\bm{r})-
  \phi\big(\bm{r}\big|\bm{\mu}_{n},
  \tilde{\bm{I}}_{0}^{-1}\big) \big|\intd\bm{r}\intd\bm{t} \\
& \quad +\int\int\left|\pi^{*}_{n,2}(\bm{t}|\bm{r})-\phi\left(\bm{t}\middle|-\big(\bm{I}^{11}_{0}\big)^{-1}\bm{I}^{12}_{0}\bm{r},\big(\bm{I}^{11}_{0}\big)^{-1}\right)\right|\phi\big(\bm{r}\big|\bm{\mu}_{n},
  \tilde{\bm{I}}_{0}^{-1}\big)\intd\bm{r}\intd\bm{t},
\end{align*}
where the first integral equals $\int\big|\pi^{*}_{n}(\bm{r})-
  \phi\big(\bm{r}\big|\bm{\mu}_{n},
  \tilde{\bm{I}}^{-1}\big) \big|\intd\bm{r}$ which goes to
zero by (\ref{eq:knot-normality-coro}). Hence showing
\begin{align}
\label{eq:tn-target-step1}
\lim_{n\to\infty}\int\int\left|\pi^{*}_{n,2}(\bm{t}|\bm{r})-\phi\left(\bm{t}\middle|-\big(\bm{I}^{11}_{0}\big)^{-1}\bm{I}^{12}_{0}\bm{r},\big(\bm{I}^{11}_{0}\big)^{-1}\right)\right|\phi\big(\bm{r}\big|\bm{\mu}_{n},
  \tilde{\bm{I}}^{-1}\big)\intd\bm{r}\intd\bm{t}=0
\end{align}
would be enough for proving (\ref{eq:joint-normality}).

\noindent{\bf Step 2}~~~~~~~~For $\delta_{r}>0$, the integral in (\ref{eq:tn-target-step1}) can be written as
\begin{align*}
&\int\int_{\bm{r}\in A_{n,1}(\delta_{r})}\left|\pi^{*}_{n,2}(\bm{t}|\bm{r})-\phi\left(\bm{t}\middle|-\big(\bm{I}^{11}_{0}\big)^{-1}\bm{I}^{12}_{0}\bm{r},\big(\bm{I}^{11}_{0}\big)^{-1}\right)\right|\phi\big(\bm{r}\big|\bm{\mu}_{n},
  \tilde{\bm{I}}^{-1}\big)\intd\bm{r}\intd\bm{t}\\
& \quad + \int\int_{\bm{r}\in A_{n,2}(\delta_{r})}\left|\pi^{*}_{n,2}(\bm{t}|\bm{r})-\phi\left(\bm{t}\middle|-\big(\bm{I}^{11}_{0}\big)^{-1}\bm{I}^{12}_{0}\bm{r},\big(\bm{I}^{11}_{0}\big)^{-1}\right)\right|\phi\big(\bm{r}\big|\bm{\mu}_{n},
  \tilde{\bm{I}}^{-1}\big)\intd\bm{r}\intd\bm{t},
\end{align*}
where the second intergal is clearly bounded by
\begin{align*}
&\int_{\bm{r}\in A_{n,2}(\delta_{r})}\int\left|\pi^{*}_{n,2}(\bm{t}|\bm{r})-\phi\left(\bm{t}\middle|-\big(\bm{I}^{11}_{0}\big)^{-1}\bm{I}^{12}_{0}\bm{r},\big(\bm{I}^{11}_{0}\big)^{-1}\right)\right|\intd\bm{t}\phi\big(\bm{r}\big|\bm{\mu}_{n},
  \tilde{\bm{I}}^{-1}\big)\intd\bm{r} \nonumber\\
\leq&2\int_{\bm{r}\in A_{n,2}(\delta_{r})}\phi\big(\bm{r}\big|\bm{\mu}_{n},
  \tilde{\bm{I}}^{-1}\big)\intd\bm{r}
\end{align*}
Transforming $\bm{r}$ back to $\bm{\zeta}$,
\begin{align*}
\int_{\bm{r}\in A_{n,2}(\delta_{r})}\phi\big(\bm{r}\big|\bm{\mu}_{n},
  \tilde{\bm{I}}^{-1}\big)\intd\bm{r}
=&\Phi\left(\|\bm{\zeta}-\bm{\zeta}_{0}\|>\delta_{r}\middle|\tilde{\bm{\zeta}}_{n}, \tilde{\bm{I}}_0/n\right).
\end{align*}
From condition~\ref{cond:1-5}, $\lim_{n\to\infty}\hat{\bm{\zeta}}_{n}= \bm{\zeta}_{0}$ and hence by continuous mapping theorem,
\begin{align*}
       &\lim_{n\to\infty}\Phi\left(\|\bm{\zeta}-\bm{\zeta}_{0}\|>\delta_{r}\middle|\tilde{\bm{\zeta}}_{n}, \tilde{\bm{I}}_0/n\right)=\lim_{n\to\infty}\Phi\left(\|\bm{\zeta}-\bm{\zeta}_{0}\|>\delta_{r}\middle|\bm{\zeta}_{0}, \tilde{\bm{I}}_0/n\right)=0.
\end{align*}
This implies that showing
\begin{align}
\label{eq:tn-target-step2}
\int\int_{\bm{r}\in A_{n,1}(\delta_{r})}\left|\pi^{*}_{n,2}(\bm{t}|\bm{r})-\phi\left(\bm{t}\middle|-\big(\bm{I}^{11}_{0}\big)^{-1}\bm{I}^{12}_{0}\bm{r},\big(\bm{I}^{11}_{0}\big)^{-1}\right)\right|\phi\big(\bm{r}\big|\bm{\mu}_{n},
  \tilde{\bm{I}}^{-1}\big)\intd\bm{r}\intd\bm{t}\to 0
\end{align}
would be enough for proving (\ref{eq:tn-target-step1}).

\noindent{\bf Step 3}~~~~~~~~
Applying (\ref{eq: lem-bi-dist-dct-tozero}) in Lemma~\ref{lem:bi-dist-dct}, we have
\begin{align*}
&\int_{\bm{t}\in B_{n,2}(\delta_{t})}\int_{\bm{r}\in A_{n,1}(\delta_{r})}\pi^{*}_{n,2}(\bm{t}|\bm{r})\phi\big(\bm{r}\big|\bm{\mu}_{n},
  \tilde{\bm{I}}_{0}^{-1}\big)\intd\bm{r}\intd\bm{t} < & \frac{|4\pi\bm{I}_{0}^{11}|^{-1/2}}{c(\delta_{r},\delta_{t})}\exp\{-n\epsilon\} \to  0.
\end{align*}
Moreover, it can be easily shown that
\begin{align*}
\int_{\bm{t}\in B_{n,2}(\delta_{t})}\int_{\bm{r}\in A_{n,1}(\delta_{r})}\phi\left(\bm{t}\middle|-\big(\bm{I}^{11}_{0}\big)^{-1}\bm{I}^{12}_{0}\bm{r},\big(\bm{I}^{11}_{0}\big)^{-1}\right)\phi\big(\bm{r}\big|\bm{\mu}_{n},
  \tilde{\bm{I}}_{0}^{-1}\big)\intd\bm{r}\intd\bm{t} \to & 0.
\end{align*}
Hence we have shown that
\begin{align*}
& \lim_{n\to\infty}\int_{\bm{t}\in B_{n,2}(\delta_{t})}\int_{\bm{r}\in
  A_{n,1}(\delta_{r})}\left|\pi^{*}_{n,2}(\bm{t}|\bm{r})-\phi\left(\bm{t}\middle|-\big(\bm{I}^{11}_{0}\big)^{-1}\bm{I}^{12}_{0}\bm{r},\big(\bm{I}^{11}_{0}\big)^{-1}\right)\right|\\
&\qquad \times\phi\big(\bm{r}\big|\bm{\mu}_{n},
  \tilde{\bm{I}}^{-1}\big)\intd\bm{r}\intd\bm{t}=0.
\end{align*}
This implies that showing
\begin{align}
\label{eq:tn-target-step3}
\begin{split}
& \lim_{n\to\infty}\int_{\bm{t}\in B_{n,1}(\delta_{t})}\int_{\bm{r}\in
  A_{n,1}(\delta_{r})}\left|\pi^{*}_{n,2}(\bm{t}|\bm{r})-\phi\left(\bm{t}\middle|-\big(\bm{I}^{11}_{0}\big)^{-1}\bm{I}^{12}_{0}\bm{r},\big(\bm{I}^{11}_{0}\big)^{-1}\right)\right|\\
&\qquad \times\phi\big(\bm{r}\big|\bm{\mu}_{n},
  \tilde{\bm{I}}_{0}^{-1}\big)\intd\bm{r}\intd\bm{t}=0
\end{split}
\end{align}
would be enough for proving (\ref{eq:tn-target-step2}).

\noindent{\bf Step 4}~~~~~~~~Since $\phi\big(\bm{r}\big|\bm{\mu}_{n},
  \tilde{\bm{I}}_{0}^{-1}\big)$ is bounded by
  $|\tilde{\bm{I}}/2\pi|^{1/2}$, applying (\ref{eq:
    lem-bi-dist-dct-limit}) in 
  Lemma~\ref{lem:bi-dist-dct} again we have
\begin{align*}
& \lim_{n\to\infty}\left|\pi^{*}_{n,2}(\bm{t}|\bm{r})-\phi\left(\bm{t}\middle|-\big(\bm{I}^{11}_{0}\big)^{-1}\bm{I}^{12}_{0}\bm{r},\big(\bm{I}^{11}_{0}\big)^{-1}\right)\right|\\
&\qquad \times\phi\big(\bm{r}\big|\bm{\mu}_{n},
  \tilde{\bm{I}}_{0}^{-1}\big) \mathbbm{1}_{\bm{t}\in B_{n,1}(\delta_{t})}\mathbbm{1}_{\bm{r}\in
  A_{n,1}(\delta_{r})}=0.
\end{align*}
Moreover, 
\begin{align*}
& \int_{\bm{t}\in B_{n,1}(\delta_{t})}\int_{\bm{r}\in
  A_{n,1}(\delta_{r})}\left|\pi^{*}_{n,2}(\bm{t}|\bm{r})-\phi\left(\bm{t}\middle|-\big(\bm{I}^{11}_{0}\big)^{-1}\bm{I}^{12}_{0}\bm{r},\big(\bm{I}^{11}_{0}\big)^{-1}\right)\right|\\
&\qquad \times\phi\big(\bm{r}\big|\bm{\mu}_{n},
  \tilde{\bm{I}}_{0}^{-1}\big)\intd\bm{r}\intd\bm{t} \\
< & \int_{\bm{t}\in B_{n,1}(\delta_{t})}\int_{\bm{r}\in
  A_{n,1}(\delta_{r})}\phi\left(\bm{t}\middle|-\big(\bm{I}^{11}_{0}\big)^{-1}\bm{I}^{12}_{0}\bm{r},\big(\bm{I}^{11}_{0}\big)^{-1}\right)\phi\big(\bm{r}\big|\bm{\mu}_{n},\tilde{\bm{I}}_{0}^{-1}\big)\intd\bm{r}\intd\bm{t}
  \\
& \qquad + \int_{\bm{t}\in B_{n,1}(\delta_{t})}\int_{\bm{r}\in
  A_{n,1}(\delta_{r})}\pi^{*}_{n,2}(\bm{t}|\bm{r})\phi\big(\bm{r}\big|\bm{\mu}_{n},
  \tilde{\bm{I}}_{0}^{-1}\big)\intd\bm{r}\intd\bm{t} \\
\leq & 2.
\end{align*}
Therefore by Scheff\'e's lemma we have shown
(\ref{eq:tn-target-step3}).

\subsection{Proof of Lemma~\ref{lem:mle-dif}}
We now provide a lemma that will be needed in our later proof of
Corollary~\ref{coro:joint-consistency}. It characterizes the asymptotic 
difference between $\tilde{\bm{\zeta}}_{n}$ and $\hat{\bm{\zeta}}_{n}$.
\begin{lemma} 
\label{lem:mle-dif}
Suppose that the conditions for Theorem~\ref{thm:joint-normality}
hold, then
\begin{align*}
\sqrt{n}\big(
\tilde{\bm{\zeta}}_{n} - 
\hat{\bm{\zeta}}_{n}
\big) \overset{d}{\to} \mathcal{N}(\bm{0},\bm{V}),
\end{align*}
where $\bm{V}$ is a positive definite matrix.
\end{lemma}

\begin{proof}
Expanding Taylor series for $\frac{1}{n}L_{n}^{(1)}(\hat{\bm{\eta}}_{n},
  \hat{\bm{\zeta}}_{n})$ we have
\begin{align*}
\bm{0}=\frac{1}{n}L_{n}^{(1)}(\hat{\bm{\eta}}_{n},
  \hat{\bm{\zeta}}_{n}) =
  \frac{1}{n}L_{n}^{(1)}(\bm{\eta}_{0},\bm{\zeta}_{0})
  +
  \frac{1}{n}L_{n}^{(2)}(\bm{\eta}_{n}^{'},\bm{\zeta}_{n}^{'})
\begin{brsm}
    \hat{\bm{\eta}}_{n}-\bm{\eta}_{0} \\
    \hat{\bm{\zeta}}_{n} - \bm{\zeta}_{0}
\end{brsm},
\end{align*}
where $\bm{\eta}_{n}^{'}$ is between $\bm{\eta}_{0}$ and
$\hat{\bm{\eta}}_{n}$ and $\bm{\zeta}_{n}^{'}$ is between $\bm{\zeta}_{0}$ and
$\hat{\bm{\zeta}}_{n}$. Letting
$\hat{\bm{I}}_{n}=\frac{1}{n}L_{n}^{(2)}(\bm{\eta}_{n}^{'},\bm{\zeta}_{n}^{'})$,
we have
\begin{align}
\label{eq:theta12}
\begin{brsm}
    \sqrt{n}\left(\hat{\bm{\eta}}_{n}-\bm{\eta}_{0}\right) \\
    \sqrt{n}\left(\hat{\bm{\zeta}}_{n} - \bm{\zeta}_{0}\right)
\end{brsm}
=
  -\hat{\bm{I}}_{n}^{-1}\sqrt{n}\frac{1}{n}L_{n}^{(1)}(\bm{\eta}_{0},\bm{\zeta}_{0}).
\end{align}
Since condition~\ref{cond:1-5} holds for
$f_{2}(\bm{y}|\bm{\eta},\bm{\zeta})$, we know that $\lim_{n\to\infty}\hat{\bm{\eta}}_{n}=\bm{\eta}_{0}$ and $\lim_{n\to\infty}\hat{\bm{\zeta}}_{n}=
\bm{\zeta}_{0}$. Applying continuous mapping theorem and condition~\ref{cond:1-3}, we have that
$\lim_{n\to\infty}\hat{\bm{I}}_{n}=\bm{I}_{0}$.
Letting
$\bm{D}_{1}$ denote a $(d_{\eta}+d_{\zeta})$-dimensional identity matrix, we
can rewrite (\ref{eq:theta12}) as
\begin{align*}
\begin{brsm}
    \sqrt{n}\left(\hat{\bm{\eta}}_{n}-\bm{\eta}_{0}\right) \\
    \sqrt{n}\left(\hat{\bm{\zeta}}_{n} - \bm{\zeta}_{0}\right)
\end{brsm}
= &
  -\hat{\bm{I}}_{n}^{-1}\bm{I}_{0}\sqrt{n}\frac{1}{n}\bm{I}_{0}^{-1}L_{n}^{(1)}(\bm{\eta}_{0},\bm{\zeta}_{0})
  \\
= &
    -\big(\hat{\bm{I}}_{n}^{-1}\bm{I}_{0}-\bm{D}_{1}\big)\sqrt{n}\frac{1}{n}\bm{I}_{0}^{-1}L_{n}^{(1)}(\bm{\eta}_{0},\bm{\zeta}_{0})
    - \sqrt{n}\frac{1}{n}\bm{I}_{0}^{-1}L_{n}^{(1)}(\bm{\eta}_{0},\bm{\zeta}_{0})
\end{align*}
Before proceeding, we introduce the following notation. Letting $\bm{A}$
be any  $(d_{\eta}+d_{\zeta})$-dimensional square matrix, we let
\begin{align*}
\bm{A} = 
\begin{bmatrix}
    \mbox{upper}(\bm{A})\\
    \mbox{lower}(\bm{A})
\end{bmatrix},
\end{align*}
where $\mbox{upper}(\bm{A})$ is a
$d_{\eta}\times (d_{\eta}+d_{\zeta})$-dimensional matrix and $\mbox{lower}(\bm{A})$ is a
$d_{\zeta}\times (d_{\eta}+d_{\zeta})$ matrix. Using this notation,
\begin{align}
\label{eq:theta122}
\begin{split}
\sqrt{n}\big(\hat{\bm{\zeta}}_{n} - \bm{\zeta}_{0}\big)
= &
    -\mbox{lower}\big[\big(\hat{\bm{I}}_{n}^{-1}\bm{I}_{0}-\bm{D}_{1}\big)\big]\sqrt{n}\frac{1}{n}\bm{I}_{0}^{-1}L_{n}^{(1)}(\bm{\eta}_{0},\bm{\zeta}_{0})\\
&\quad  -
\mbox{lower}\big(\bm{D}_{1}\big)\sqrt{n}\frac{1}{n}\bm{I}_{0}^{-1}L_{n}^{(1)}(\bm{\eta}_{0},\bm{\zeta}_{0})
\end{split}
\end{align}

Expanding Taylor series for $\frac{1}{n}Q_{n}^{(1)}(\tilde{\bm{\zeta}}_{n})$, we have
\begin{align*}
\bm{0}=\frac{1}{n}Q_{n}^{(1)}(\tilde{\bm{\zeta}}_{n}) =
\frac{1}{n}Q_{n}^{(1)}(\bm{\zeta}_{0})
  +
  \frac{1}{n}Q_{n}^{(2)}(\bm{\zeta}_{n}^{*})\big(\tilde{\bm{\zeta}}_{n} - \bm{\zeta}_{0}\big),
\end{align*}
where $\bm{\zeta}_{n}^{*}$ is between $\bm{\zeta}_{0}$ and
$\tilde{\bm{\zeta}}_{n}$. Letting
$\tilde{\bm{I}}_{n}=\frac{1}{n}Q_{n}^{(2)}(\bm{\zeta}_{n}^{*})$,
we have
\begin{align}
\label{eq:zeta}
\sqrt{n}\big(\tilde{\bm{\zeta}}_{n} - \bm{\zeta}_{0}\big)=
-\tilde{\bm{I}}_{n}^{-1}\sqrt{n}\frac{1}{n}Q_{n}^{(1)}(\bm{\zeta}_{0}).
\end{align}
Similarly, using the fact that
conditions~\ref{cond:1-3},~\ref{cond:1-5} hold for
$f_{1}(\bm{y}|\bm{\zeta})$ and the
continuous mapping theorem, we have that $\lim_{n\to\infty}\tilde{\bm{I}}_{n} = \tilde{\bm{I}}_{0}$.
Letting $\bm{D}_{2}$ denote the $d_{\zeta}$-dimensional identity
matrix, we can rewrite (\ref{eq:zeta}) as
\begin{align*}
\begin{split}
\sqrt{n}\big(\tilde{\bm{\zeta}}_{n} - \bm{\zeta}_{0}\big)= &
-\tilde{\bm{I}}_{n}^{-1}\tilde{\bm{I}}_0\sqrt{n}\frac{1}{n}\tilde{\bm{I}}_0^{-1}Q_{n}^{(1)}(\bm{\zeta}_{0})
  \\
 = &
     -\big(\tilde{\bm{I}}_{n}^{-1}\tilde{\bm{I}}_0-\bm{D}_{2}\big)\sqrt{n}\frac{1}{n}\tilde{\bm{I}}^{-1}Q_{n}^{(1)}(\bm{\zeta}_{0})
     -
     \sqrt{n}\frac{1}{n}\tilde{\bm{I}}_0^{-1}Q_{n}^{(1)}(\bm{\zeta}_{0})
\end{split}
\end{align*}

Combining above and (\ref{eq:theta122}), we have
\begin{align*}
 & \sqrt{n}\big(
\tilde{\bm{\zeta}}_{n} - 
\hat{\bm{\zeta}}_{n}
\big) \\
 = &
     -\big(\tilde{\bm{I}}_{n}^{-1}\tilde{\bm{I}}_0-\bm{D}_{2}\big)\sqrt{n}\frac{1}{n}\tilde{\bm{I}}_{0}^{-1}Q_{n}^{(1)}(\bm{\zeta}_{0}) +\mbox{lower}\big[\big(\hat{\bm{I}}_{n}^{-1}\bm{I}_{0}-\bm{D}_{1}\big)\big]\sqrt{n}\frac{1}{n}\bm{I}_{0}^{-1}L_{n}^{(1)}(\bm{\eta}_{0},\bm{\zeta}_{0}) \\
& \quad -
     \sqrt{n}\frac{1}{n}\bigg[\tilde{\bm{I}}_{0}^{-1}Q_{n}^{(1)}(\bm{\zeta}_{0})-\mbox{lower}\big(\bm{D}_{1}\big)\bm{I}_{0}^{-1}L_{n}^{(1)}(\bm{\eta}_{0},\bm{\zeta}_{0})\bigg].
\end{align*}

Since condition~\ref{cond:1-3} holds for
$f_{1}(\bm{y}|\bm{\zeta})$, we know that
$\mathbb{E}_{\bm{\theta}_{0}}\triangledown_{\bm{\zeta}} \ell_{1}(\bm{\zeta},
\bm{y})|_{\bm{\theta}=\bm{\theta}_{0}} = \bm{0}$,
and hence 
$\mbox{Var}_{\bm{\theta}_0} \triangledown_{\bm{\zeta}} \ell_{1}(\bm{\zeta},
\bm{y})|_{\bm{\theta}=\bm{\theta}_{0}} = \tilde{\bm{I}}_0$.
Similarly, since condition~\ref{cond:1-3} holds for
$f_{2}(\bm{y}|\bm{\eta},\bm{\zeta})$, we know that
$\mathbb{E}_{\bm{\theta}_{0}}\triangledown_{\bm{\eta},\bm{\zeta}} \ell_{2}(\bm{\eta}, \bm{\zeta},
\bm{y}) |_{\bm{\theta}=\bm{\theta}_{0}} = 0$ and
$\mbox{Var}_{\bm{\theta}_0} \triangledown_{\bm{\eta},\bm{\zeta}} \ell_{2}(\bm{\eta}, \bm{\zeta},
\bm{y}) |_{\bm{\theta}=\bm{\theta}_{0}} = \bm{I}_0$.
It is easily seen that
\begin{align*}
&
  \mathbb{E}_{\bm{\theta}_0}\tilde{\bm{I}}_{0}^{-1} \triangledown_{\bm{\zeta}} \ell_{1}(\bm{\zeta},
\bm{y})|_{\bm{\theta}=\bm{\theta}_{0}}=\bm{0}, \quad \mathbb{E}_{\bm{\theta}_0}\mbox{lower}\big(\bm{D}_{1}\big)\bm{I}_{0}^{-1}\triangledown_{\bm{\eta},\bm{\zeta}} \ell_{2}(\bm{\eta}, \bm{\zeta},
\bm{y}) |_{\bm{\theta}=\bm{\theta}_{0}}=\bm{0}
\end{align*}
and hence
\begin{align*}
&\mathbb{E}_{\bm{\theta}_0}\left[\tilde{\bm{I}}_{0}^{-1} \triangledown_{\bm{\zeta}} \ell_{1}(\bm{\zeta},
\bm{y})-\mbox{lower}\big(\bm{D}_{1}\big)\bm{I}_{0}^{-1}\triangledown_{\bm{\eta},\bm{\zeta}} \ell_{2}(\bm{\eta}, \bm{\zeta},
\bm{y})\right]\bigg|_{\bm{\theta}=\bm{\theta}_{0}}=\bm{0}.
\end{align*}
Also, it can be shown that
\begin{align*}
&\mbox{Var}_{\bm{\theta}_0}\tilde{\bm{I}}_{0}^{-1} \triangledown_{\bm{\zeta}} \ell_{1}(\bm{\zeta},
\bm{y})|_{\bm{\theta}=\bm{\theta}_{0}}=\tilde{\bm{I}}_{0}^{-1}
  \tilde{\bm{I}}_0 \tilde{\bm{I}}^{-1}_0 =\tilde{\bm{I}}^{-1}_0,\\
&\mbox{Var}_{\bm{\theta}_0}\bm{I}_{0}^{-1}\triangledown_{\bm{\eta},\bm{\zeta}} \ell_{2}(\bm{\eta}, \bm{\zeta},
\bm{y}) |_{\bm{\theta}=\bm{\theta}_{0}} =\bm{I}_{0}^{-1}\bm{I}_{0}\bm{I}_{0}^{-1}=\bm{I}_{0}^{-1}.
\end{align*}
Applying central limit theorem (CLT), we have
\begin{align}
& \sqrt{n}\frac{1}{n}\tilde{\bm{I}}^{-1}Q_{n}^{(1)}(\bm{\zeta}_{0})
  \overset{d}{\to} \mathcal{N}\big(\bm{0}, \tilde{\bm{I}}^{-1}\big) \label{eq:zeta-clt}\\
&
  \sqrt{n}\frac{1}{n}\bm{I}_{0}^{-1}L_{n}^{(1)}(\bm{\eta}_{0},\bm{\zeta}_{0}) \overset{d}{\to} \mathcal{N}\big(\bm{0},\bm{I}_{0}^{-1}\big). \label{eq:theta11-clt}
\end{align}
Introducing
\begin{align*}
\bm{V}=\mbox{Var}_{\bm{\theta}_{0}} \left[\tilde{\bm{I}}^{-1}_0 \triangledown_{\bm{\zeta}} \ell_{1}(\bm{\zeta},
\bm{y})-\mbox{lower}\big(\bm{D}_{1}\big)\bm{I}_{0}^{-1}\triangledown_{\bm{\eta},\bm{\zeta}} \ell_{2}(\bm{\eta}, \bm{\zeta},
\bm{y})\right]\bigg|_{\bm{\theta}=\bm{\theta}_{0}}
\end{align*}
and applying CLT again, we have
\begin{align*}
\sqrt{n}\frac{1}{n}\bigg[\tilde{\bm{I}}^{-1}_0Q_{n}^{(1)}(\bm{\zeta}_{0})-\mbox{lower}\big(\bm{D}_{1}\big)\bm{I}_{0}^{-1}L_{n}^{(1)}(\bm{\eta}_{0},\bm{\zeta}_{0})\bigg]
  \overset{d}{\to} \mathcal{N}\big(\bm{0}, \bm{V}\big).
\end{align*}
We have already shown that $\lim_{n\to\infty}\tilde{\bm{I}}_{n}=
\tilde{\bm{I}}_0$, which implies that
$\lim_{n\to\infty}\tilde{\bm{I}}_{n}^{-1}\tilde{\bm{I}}_0 =
\bm{D}_{2}$. Similarly, we have shown that $\lim_{n\to\infty}\hat{\bm{I}}_{n}=\bm{I}_0$, which implies
$\lim_{n\to\infty}\hat{\bm{I}}_{n}^{-1}\bm{I}_0 = \bm{D}_{1}$. Combining (\ref{eq:zeta-clt}) and (\ref{eq:theta11-clt}) and
applying Slutsky's theorem, we have
\begin{align*}
& \big(\tilde{\bm{I}}_{n}^{-1}\tilde{\bm{I}}_0-\bm{D}_{2}\big)\sqrt{n}\frac{1}{n}\tilde{\bm{I}}^{-1}Q_{n}^{(1)}(\bm{\zeta}_{0})
  \overset{d}{\to} \bm{0}, \\
&
  \mbox{lower}\big[\big(\hat{\bm{I}}_{n}^{-1}\bm{I}_{0}-\bm{D}_{1}\big)\big]\sqrt{n}\frac{1}{n}\bm{I}_{0}^{-1}L_{n}^{(1)}(\bm{\eta}_{0},\bm{\zeta}_{0}) \overset{d}{\to} \bm{0}.
\end{align*}
Since convergence in distribution to a
constant implies convergence in probability, we have
\begin{align*}
& \big(\tilde{\bm{I}}_{n}^{-1}\tilde{\bm{I}}_0-\bm{D}_{2}\big)\sqrt{n}\frac{1}{n}\tilde{\bm{I}}^{-1}Q_{n}^{(1)}(\bm{\zeta}_{0})
  \overset{P_{\bm{\theta}_0}}{\to} \bm{0}, \\
&
  \mbox{lower}\big[\big(\hat{\bm{I}}_{n}^{-1}\bm{I}_{0}-\bm{D}_{1}\big)\big]\sqrt{n}\frac{1}{n}\bm{I}_{0}^{-1}L_{n}^{(1)}(\bm{\eta}_{0},\bm{\zeta}_{0}) \overset{P_{\bm{\theta}_0}}{\to} \bm{0},
\end{align*}
which implies that
\begin{align*}
\sqrt{n}\big(
\tilde{\bm{\zeta}}_{n} - 
\hat{\bm{\zeta}}_{n}
\big) \overset{P_{\bm{\theta}_0}}{\to} & -
     \sqrt{n}\frac{1}{n}\bigg[\tilde{\bm{I}}^{-1}_0Q_{n}^{(1)}(\bm{\zeta}_{0})-\mbox{lower}\big(\bm{D}_{1}\big)\bm{I}_{0}^{-1}L_{n}^{(1)}(\bm{\eta}_{0},\bm{\zeta}_{0})\bigg],
\end{align*}
where the right part has been shown to converge in distribution to
$\mathcal{N}(\bm{0},\bm{V})$. Hence we can conclude that
\begin{align*}
\sqrt{n}\big(
\tilde{\bm{\zeta}}_{n} - 
\hat{\bm{\zeta}}_{n}
\big) \overset{d}{\to} & \mathcal{N}(\bm{0},\bm{V}).
\end{align*}

\end{proof}

\subsection{Proof of Corollary~\ref{coro:joint-consistency}}
Applying random variable transformation to (\ref{eq:joint-normality}),
it can be shown that 
\begin{align*}
&\lim_{n\to\infty}\int \int
  \left|g_{n}( \bm{\eta}, \bm{\zeta})\right|\intd \bm{\eta} \intd \bm{\zeta}= 0,
\end{align*}
where
\begin{align*}
& g_{n}( \bm{\eta}, \bm{\zeta})
= \tau_{n}\left(\bm{\eta}|\bm{\zeta}\right)
  \kappa_{n}\left(\bm{\zeta}\right)\\
&\quad-\phi\left(\bm{\eta}\middle|\hat{\bm{\eta}}_{n}-\big(\bm{I}^{11}_{0}\big)^{-1}\bm{I}^{12}_{0}\left(\bm{\zeta}-\hat{\bm{\zeta}}_{n}\right),\big(\bm{I}^{11}_{0}\big)^{-1}/\sqrt{n}\right) \phi\left(\bm{\zeta}\middle|\tilde{\bm{\zeta}}_{n},
  \tilde{\bm{I}}_{0}^{-1}/\sqrt{n}\right).
\end{align*}
Let $\bm{z}=\begin{brsm}
\bm{\eta} \\
    \bm{\zeta}
\end{brsm}$ and $\bm{z}_{0}=\begin{brsm}
\bm{\eta}_{0} \\
    \bm{\zeta}_{0}
\end{brsm}$. For any neighborhood $U$ of $\bm{z}_{0}$, $\exists\delta>0$
such that $B=\mathcal{B}_{\bm{z}}(\bm{z}_{0},\delta)\in U$. Then it can be shown
that
\begin{align*}
 & \lim_{n\to\infty}\int_{U}\tau_{n}\left(\bm{\eta}|\bm{\zeta}\right)
  \kappa_{n}\left(\bm{\zeta}\right) \intd \bm{\eta} \intd \bm{\zeta}\\
= & \lim_{n\to\infty}\int_{U}\phi\left(\bm{\eta}\middle|\hat{\bm{\eta}}_{n}-\big(\bm{I}^{11}_{0}\big)^{-1}\bm{I}^{12}_{0}\left(\bm{\zeta}-\hat{\bm{\zeta}}_{n}\right),\big(\bm{I}^{11}_{0}\big)^{-1}/\sqrt{n}\right) \phi\left(\bm{\zeta}\middle|\tilde{\bm{\zeta}}_{n},
  \tilde{\bm{I}}_{0}^{-1}/\sqrt{n}\right)\intd \bm{\eta} \intd
    \bm{\zeta} \\
& \qquad + \lim_{n\to\infty}\int_{U}
\left|g_{n}( \bm{\eta}, \bm{\zeta})\right|\intd \bm{\eta} \intd \bm{\zeta}\\
= & \lim_{n\to\infty}\int_{U}\phi\left(\bm{\eta}\middle|\hat{\bm{\eta}}_{n}-\big(\bm{I}^{11}_{0}\big)^{-1}\bm{I}^{12}_{0}\left(\bm{\zeta}-\hat{\bm{\zeta}}_{n}\right),\big(\bm{I}^{11}_{0}\big)^{-1}/\sqrt{n}\right) \phi\left(\bm{\zeta}\middle|\tilde{\bm{\zeta}}_{n},
  \tilde{\bm{I}}_{0}^{-1}/\sqrt{n}\right)\intd \bm{\eta} \intd
    \bm{\zeta} \\
\geq & \lim_{n\to\infty}\int_{B}\phi\left(\bm{\eta}\middle|\hat{\bm{\eta}}_{n}-\big(\bm{I}^{11}_{0}\big)^{-1}\bm{I}^{12}_{0}\left(\bm{\zeta}-\hat{\bm{\zeta}}_{n}\right),\big(\bm{I}^{11}_{0}\big)^{-1}/\sqrt{n}\right) \phi\left(\bm{\zeta}\middle|\tilde{\bm{\zeta}}_{n},
  \tilde{\bm{I}}_{0}^{-1}/\sqrt{n}\right)\intd \bm{\eta} \intd
    \bm{\zeta} 
\end{align*}
Since condition~\ref{cond:1-5} holds for both
$f_{1}(\bm{y}|\bm{\zeta})$ and $f_{2}(\bm{y}|\bm{\eta},\bm{\zeta})$, we have
$\lim_{n\to\infty}\hat{\bm{\eta}}_{n}= \bm{\eta}_{0}$,
$\lim_{n\to\infty}\hat{\bm{\zeta}}_{n}= \bm{\zeta}_{0}$ and $\lim_{n\to\infty}\tilde{\bm{\zeta}}_{n}= \bm{\zeta}_{0}$. Hence the above limit goes to 1.

\subsection{Proof of Theorem~\ref{thm:plugin-tile-normality}}
Letting
$\bm{r}^{*}_{n}=\sqrt{n}\left(\bm{\zeta}^{*}_{n}-\hat{\bm{\zeta}}_{n}\right)$,
it can be seen that $\pi_{n,3}^{*}(\bm{t})=\pi^{*}_{n,2}(\bm{t}|\bm{r}^{*}_{n})$.
Suppose that the conditions for Theorem~\ref{thm:joint-normality}
hold, then slightly modifying the step 3 and step 4 in the proof of
Theorem~\ref{thm:joint-normality} we can show that $\exists \delta>0$ such that for
$\left\|\bm{r}^{*}_{n}\right\|<\sqrt{n}\delta$,
\begin{align*}
\lim_{n\to\infty}\int \left|
  \pi_{n,3}^{*}(\bm{t})-\phi\left(\bm{t}\middle|-\big(\bm{I}^{11}_{0}\big)^{-1}\bm{I}^{12}_{0}\bm{\mu}_{n},\big(\bm{I}^{11}_{0}\big)^{-1}\right)
\right|\intd \bm{t}= 0.
\end{align*}
Define a sequence of events
$A_{n,\delta}=\{\bm{r}^{*}_{n}:\left\|\bm{r}^{*}_{n}\right\|<\sqrt{n}\delta\}$ and
$B_{n,\delta}=\{\bm{r}^{*}_{n}:\left\|\bm{r}^{*}_{n}\right\|>\sqrt{n}\delta\}$, then for any $\epsilon>0$, we have
\begin{align} \label{eq:plugin-proof}
\begin{split}
& p\left[\int \left|
  \pi_{n,3}^{*}(\bm{t})-\phi\left(\bm{t}\middle|-\big(\bm{I}^{11}_{0}\big)^{-1}\bm{I}^{12}_{0}\bm{\mu}_{n},\big(\bm{I}^{11}_{0}\big)^{-1}\right)
\right|\intd \bm{t}>\epsilon\right] \\
=& p(A_{n,\delta}) p\left[\int \left|
  \pi_{n,3}^{*}(\bm{t})-\phi\left(\bm{t}\middle|-\big(\bm{I}^{11}_{0}\big)^{-1}\bm{I}^{12}_{0}\bm{\mu}_{n},\big(\bm{I}^{11}_{0}\big)^{-1}\right)
\right|\intd \bm{t}>\epsilon\bigg|A_{n,\delta}\right] \\
&\quad+p(B_{n,\delta}) p\left[\int \left|
  \pi_{n,3}^{*}(\bm{t})-\phi\left(\bm{t}\middle|-\big(\bm{I}^{11}_{0}\big)^{-1}\bm{I}^{12}_{0}\bm{\mu}_{n},\big(\bm{I}^{11}_{0}\big)^{-1}\right)
\right|\intd \bm{t}>\epsilon\bigg|B_{n,\delta}\right]
\end{split}
\end{align}
We have already shown that
\begin{align*}
\lim_{n\to\infty}p\left[\int \left|
  \pi_{n,3}^{*}(\bm{t})-\phi\left(\bm{t}\middle|-\big(\bm{I}^{11}_{0}\big)^{-1}\bm{I}^{12}_{0}\bm{\mu}_{n},\big(\bm{I}^{11}_{0}\big)^{-1}\right)
\right|\intd \bm{t}>\epsilon\bigg|A_{n,\delta}\right]=0.
\end{align*}
Since $p(A_{n,1})\leq 1$, the first part in
(\ref{eq:plugin-proof}) goes to 0.

Similarly, we also know
\[p\left(\int \bigg|
  \pi_{n,3}^{*}(\bm{t})-\phi\bigg(\bm{t}\bigg|-\big(\bm{I}^{11}_{0}\big)^{-1}\bm{I}^{12}_{0}\bm{\mu}_{n},\big(\bm{I}^{11}_{0}\big)^{-1}\bigg)
\bigg|d \bm{t}>\epsilon\bigg|A_{n,2}\right)\leq 1.\]
From Lemma~\ref{lem:knot-posterior-mean}, we know that
$\lim_{n\to\infty}\sqrt{n}\left(\bm{\zeta}^{*}_{n}- \tilde{\bm{\zeta}}_{n}\right)=0$. From
Lemma~\ref{lem:mle-dif}, we know that $\sqrt{n}\big(
\tilde{\bm{\zeta}}_{n} - 
\hat{\bm{\zeta}}_{n}
\big) \overset{d}{\to} \mathcal{N}(\bm{0},\bm{V})$. Combining them we
can show that $\bm{r}^{*}_{n} \overset{d}{\to}
\mathcal{N}(\bm{0},\bm{V})$, hence 
\[p(B_{n,\delta}) =p(\left\|\bm{r}^{*}_{n}\right\|>\sqrt{n}\delta)\]
also goes to zero. We have shown that the second part in
(\ref{eq:plugin-proof}) goes to 0.

\subsection{Proof of Lemma~\ref{coro:knot-complexity-bound}}
We first provide the following lemma.
\begin{lemma} \label{lem:normal-evt}
(\citet{david1970order}) Let $\Phi$ be the cdf function of the standard
normal distribution, then
\begin{align} \label{eq:stdn-extreme}
\lim_{n\to\infty}\Phi\left(a_n x+b_n\right)^n=e^{-\exp(-x)},
\end{align}
where $b_{n} = \Phi^{-1}\left(1-\frac{1}{n}\right)$ and $a_{n} = \frac{1}{n\phi(b_{n})}$.
\end{lemma}

For any $j$ and for fixed $\mu_{j}$ and $\sigma_{jj}$, we standardized the $x_{ij}$'s and introduce
$z_{i}=\frac{x_{ij}-\mu_{j}}{\sqrt{\sigma_{jj}}}$. Clearly
$z_{1},\ldots,z_{n}$ are i.i.d. standard normal random
variables. Applying Lemma~\ref{lem:normal-evt}, for any $\delta>0$,
\[\lim_{n\to\infty}p(\max_{1\leq i\leq n}z_{i}<a_{n}\delta+b_{n})=
  e^{-\exp(-\delta)}.\]
Transforming $z_{i}$'s back to $x_{ij}$'s, we get
\[\lim_{n\to\infty}p\left[\max_{1\leq i\leq n}x_{ij}<\sqrt{\sigma_{jj}}\left(a_{n}\delta+b_{n}\right)+\mu_{j}\right]=
  e^{-\exp(-\delta)}.\]

Since $y_{ij}=\mathbbm{1}\{x_{j}>0\}\ceil{x_{ij}}$, it is easily seen
that for any $a\geq 0$, $\max_{1\leq i\leq n}x_{ij}<a$ implies $\max_{1\leq i\leq n}y_{ij}<a+1$. Hence
\[p\left[\max_{1\leq i\leq
      n}y_{ij}<\sqrt{\sigma_{jj}}\left(a_{n}\delta+b_{n}\right)+\mu_{j}+1\right]
\geq
  p\left[\max_{1\geq i\leq
      n}x_{ij}<\sqrt{\sigma_{jj}}\left(a_{n}\delta+b_{n}\right)+\mu_{j}\right]\]
holds for any $n$, which implies
\begin{align*}
\lim_{n\to\infty}pr\left[\max_{1\leq i\leq
  n}y_{ij}<\sqrt{\sigma_{jj}}\left(a_{n}\delta+b_{n}\right)+\mu_{j}+1\right]\geq e^{-\exp(-\delta)}.
\end{align*}

It is easily seen that $b_n\to\infty$. Using Mills ratio, we can show
that for any $x>0$,
$\lim_{n\to\infty}\frac{1-\Phi(b_{n})}{\phi(b_{n})}=\frac{1}{b_{n}}$.
Noting that $1-\Phi(b_{n})=\frac{1}{n}$, we have shown that
\begin{align} \label{eq:mills-ratio}
\lim_{n\to\infty}a_{n}=\frac{1}{b_{n}}.
\end{align}

Integrating by parts, one can easily show the following two bounds:
\begin{align*}
1-\Phi(x) \leq \frac{e^{-x^2/2}}{\sqrt{2\pi x}}, & \quad 1-\Phi(x) \geq \frac{e^{-x^2/2}}{\sqrt{2\pi}}\left(\frac{1}{x}-\frac{1}{x^{3}}\right).
\end{align*}
Since $1-\Phi(b_{n})=\frac{1}{n}$, it can be shown that
$\sqrt{2\log n} \leq b_{n}\leq \sqrt{2\log n}$ for sufficiently large
$n$. Coupled with (\ref{eq:mills-ratio}),
we would have $a_{n}<\frac{2}{\sqrt{\log n}}$. This implies that
\[a_{n}\delta+b_{n} < \frac{2\delta}{\sqrt{\log n}}+\sqrt{2\log n},\]
and hence
\begin{align*}
& pr\left[\max_{1\leq i\leq
  n}y_{ij}<\sqrt{\sigma_{jj}}\left(a_{n}\delta+b_{n}\right)+\mu_{j}+1\right] \\
< & pr\left[\max_{1\leq i\leq
  n}y_{ij}<\sqrt{\sigma_{jj}}\left(\frac{2\delta}{\sqrt{\log
      n}}+\sqrt{2\log n}\right)+\mu_{j}+1\right],
\end{align*}
which completes the proof.

\newpage
\bibliography{references}

\begin{thebibliography}{17}
\providecommand{\natexlab}[1]{#1}
\providecommand{\url}[1]{\texttt{#1}}
\expandafter\ifx\csname urlstyle\endcsname\relax
  \providecommand{\doi}[1]{doi: #1}\else
  \providecommand{\doi}{doi: \begingroup \urlstyle{rm}\Url}\fi

\bibitem[Attias(2000)]{attias2000variational}
H.~Attias.
\newblock A variational baysian framework for graphical models.
\newblock In \emph{Advances in neural information processing systems}, pages
  209--215, 2000.

\bibitem[Canale and Dunson(2011)]{canale2011bayesian}
A.~Canale and D.~B. Dunson.
\newblock Bayesian kernel mixtures for counts.
\newblock \emph{Journal of the American Statistical Association}, 106\penalty0
  (496):\penalty0 1528--1539, 2011.

\bibitem[Cox and Reid(2004)]{cox2004note}
D.~R. Cox and N.~Reid.
\newblock A note on pseudolikelihood constructed from marginal densities.
\newblock \emph{Biometrika}, 91\penalty0 (3):\penalty0 729--737, 2004.

\bibitem[David and Nagaraja(1970)]{david1970order}
H.~A. David and H.~N. Nagaraja.
\newblock \emph{Order statistics}.
\newblock Wiley Online Library, 1970.

\bibitem[El-Basyouny et~al.(2014)El-Basyouny, Barua, and
  Islam]{el2014investigation}
K.~El-Basyouny, S.~Barua, and M.~T. Islam.
\newblock Investigation of time and weather effects on crash types using full
  {B}ayesian multivariate {P}oisson lognormal models.
\newblock \emph{Accident Analysis \& Prevention}, 73:\penalty0 91--99, 2014.

\bibitem[Folland(2005)]{folland2005higher}
G.~Folland.
\newblock Higher-order derivatives and taylor's formula in several variables,
  2005.

\bibitem[Gelman et~al.(2006)]{gelman2006prior}
A.~Gelman et~al.
\newblock Prior distributions for variance parameters in hierarchical models
  (comment on article by browne and draper).
\newblock \emph{Bayesian analysis}, 1\penalty0 (3):\penalty0 515--534, 2006.

\bibitem[Ghosh et~al.(2007)Ghosh, Delampady, and
  Samanta]{ghosh2007introduction}
J.~K. Ghosh, M.~Delampady, and T.~Samanta.
\newblock \emph{An introduction to {B}ayesian analysis: theory and methods}.
\newblock Springer Science \& Business Media, 2007.

\bibitem[Jaakkola and Jordan(2000)]{jaakkola2000bayesian}
T.~S. Jaakkola and M.~I. Jordan.
\newblock Bayesian parameter estimation via variational methods.
\newblock \emph{Statistics and Computing}, 10\penalty0 (1):\penalty0 25--37,
  2000.

\bibitem[Johndrow et~al.(2016)Johndrow, Smith, Pillai, and
  Dunson]{johndrow2016inefficiency}
J.~E. Johndrow, A.~Smith, N.~Pillai, and D.~B. Dunson.
\newblock Inefficiency of data augmentation for large sample imbalanced data.
\newblock \emph{arXiv preprint arXiv:1605.05798}, 2016.

\bibitem[Lewandowski et~al.(2009)Lewandowski, Kurowicka, and
  Joe]{lewandowski2009generating}
D.~Lewandowski, D.~Kurowicka, and H.~Joe.
\newblock Generating random correlation matrices based on vines and extended
  onion method.
\newblock \emph{Journal of multivariate analysis}, 100\penalty0 (9):\penalty0
  1989--2001, 2009.

\bibitem[Ma et~al.(2008)Ma, Kockelman, and Damien]{ma2008multivariate}
J.~Ma, K.~M. Kockelman, and P.~Damien.
\newblock A multivariate {P}oisson-lognormal regression model for prediction of
  crash counts by severity, using {B}ayesian methods.
\newblock \emph{Accident Analysis \& Prevention}, 40:\penalty0 964--975, 2008.

\bibitem[Pauli et~al.(2011)Pauli, Racugno, and Ventura]{pauli2011bayesian}
F.~Pauli, W.~Racugno, and L.~Ventura.
\newblock Bayesian composite marginal likelihoods.
\newblock \emph{Statistica Sinica}, pages 149--164, 2011.

\bibitem[Rue et~al.(2009)Rue, Martino, and Chopin]{rue2009approximate}
H.~Rue, S.~Martino, and N.~Chopin.
\newblock Approximate bayesian inference for latent gaussian models by using
  integrated nested laplace approximations.
\newblock \emph{Journal of the royal statistical society: Series b (statistical
  methodology)}, 71\penalty0 (2):\penalty0 319--392, 2009.

\bibitem[Scott et~al.(2016)Scott, Blocker, Bonassi, Chipman, George, and
  McCulloch]{scott2016bayes}
S.~L. Scott, A.~W. Blocker, F.~V. Bonassi, H.~A. Chipman, E.~I. George, and
  R.~E. McCulloch.
\newblock Bayes and big data: The consensus {M}onte {C}arlo algorithm.
\newblock \emph{International Journal of Management Science and Engineering
  Management}, 11\penalty0 (2):\penalty0 78--88, 2016.

\bibitem[Tang et~al.(2008)Tang, Zhang, Yao, Li, Zhang, and
  Su]{tang2008arnetminer}
J.~Tang, J.~Zhang, L.~Yao, J.~Li, L.~Zhang, and Z.~Su.
\newblock Arnetminer: extraction and mining of academic social networks.
\newblock In \emph{Proceedings of the 14th ACM SIGKDD international conference
  on Knowledge discovery and data mining}, pages 990--998. ACM, 2008.

\bibitem[Wang and Dunson(2013)]{wang2013parallelizing}
X.~Wang and D.~B. Dunson.
\newblock Parallelizing {MCMC} via {W}eierstrass sampler.
\newblock \emph{arXiv preprint arXiv:1312.4605}, 2013.

\end{thebibliography}

 \end{document}